\documentclass{article}
\usepackage[utf8]{inputenc}
\usepackage{fullpage}
\usepackage{amsmath, amsthm, amssymb}
\usepackage[table]{xcolor}
\usepackage{bm}
\usepackage{caption, subcaption}
\usepackage[sort&compress, square, numbers]{natbib}
\usepackage{multirow, multicol, booktabs}
\usepackage{soul}
\PassOptionsToPackage{hyphens}{url}
\usepackage[bookmarks=false]{hyperref}
\usepackage{graphicx}

\usepackage{amsmath,amsfonts,amsthm}
\usepackage{thmtools} 
\usepackage{xcolor}
\usepackage{amssymb}
\usepackage{bm}
\usepackage{xspace}
\usepackage{mathtools}

\usepackage{nicefrac}
\usepackage{natbib}

\usepackage{comment}

\setcitestyle{authoryear}

\newtheorem{theorem}{Theorem}
\newtheorem{lemma}{Lemma}
\newtheorem{proposition}{Proposition}

\newtheorem{corollary}{Corollary}
\newtheorem{definition}{Definition}

\newtheorem{remark}{Remark}
\newtheorem{assumption}{Assumption}

\def\R{\mathbb{R}}
\def \P{\mathbb{P}}

\newcommand{\E}{\mathbb{E}}

\newcommand{\ua}{ \mbox{mCPA} }
\newcommand{\tcpa}{ \mbox{tCPA} }
\newcommand{\spa}{ \mbox{\tiny SPA} }
\newcommand{\mye}{ \mbox{\tiny MYE} }
\renewcommand{\b} {\bm{b}}
\newcommand{\q} {\bm{q}}
\renewcommand{\v} {\bm{v}}

\usepackage{ifthen}
\newboolean{draft}
\setboolean{draft}{true}
\ifthenelse{\boolean{draft}}{

\newcommand{\aranyak}[1]{\textcolor{orange}{Aranyak: {#1}}}
\newcommand{\andres}[1]{\textcolor{purple}{Andres: {#1}}}
\newcommand{\apnote}[1]{\textcolor{purple}{Andres: {#1}}}
}{

\newcommand{\chris}[1]{}
\newcommand{\aranyak}[1]{}
\newcommand{\andres}[1]{}
\newcommand{\apnote}[1]{}
}

\begin{document}

\title{Auctions without commitment in the auto-bidding world}

\author{
    Aranyak Mehta\thanks{Google, \texttt{\{aranyak,perlroth\}@google.com}},
    Andres Perlroth\footnotemark[1]
}



\maketitle

\begin{abstract}
Advertisers in online ad auctions are increasingly using auto-bidding mechanisms to bid into auctions instead of directly bidding their value manually. One of the prominent auto-bidding formats is that of target cost-per-acquisition ($\tcpa$) which maximizes the volume of conversions subject to a return-of-investment constraint. From an auction theoretic perspective however, this trend seems to go against foundational results that postulate that for profit-maximizing (\emph{aka} quasi-linear) bidders, it is optimal to use a classic bidding system like marginal CPA ($\ua$) bidding rather than using strategies like $\tcpa$. 

In this paper we rationalize the adoption of such seemingly sub-optimal bidding within the canonical quasi-linear framework. The crux of the argument lies in the notion of \emph{commitment}. We consider a multi-stage game where first the auctioneer declares the auction rules; then bidders select either the $\tcpa$ or $\ua$ bidding format (and submit bids accordingly); and then, if the auctioneer lacks commitment, it can revisit the rules of the auction (e.g., may readjust reserve prices depending on the bids submitted by the bidders).
Our main result is that so long as a bidder believes that the auctioneer lacks commitment to follow the rule of the declared auction then the bidder will make a higher profit by choosing the $\tcpa$ format compared to the classic $\ua$ format. 

We then explore the commitment consequences for the auctioneer. In a simplified version of the model where there is only one bidder, we show that the $\tcpa$ subgame admits a \emph{credible} equilibrium while the $\ua$ format does not. That is, when the bidder chooses the $\tcpa$ format the auctioneer can credibly implement the auction rules announced at the beginning of the game. We also show that, under some mild conditions, the auctioneer's revenue is larger when the bidder to uses the $\tcpa$ format rather than $\ua$. We further quantify the value for the auctioneer to be able to commit to the declared auction rules. 
\end{abstract}
\providecommand{\keywords}[1]
{
  \small	
  \textbf{\textbf{Keywords:}} #1
}
\keywords{Auto-bidding, Auction Design, Mechanism Design, Credible mechanisms, Equilibrium, Economics}

\section{Introduction}

Over the past several years, advertisers have increasingly started to use auto-bidding mechanisms to bid in ad auctions instead of directly bidding their value manually (e.g., bidding per keyword in sponsored search). Among the prominent bidding strategies that have been adopted is the target cost per acquisition ($\tcpa$) strategy (see, e.g.,~\citet{googleautobiddingsupport,fbautobiddingsupport}) where the goal is to maximize the volume of conversions \emph{aka} acquisitions, subject to an upper bound on cost per conversion\footnote{A related strategy, target return-on-ad-spend (tROAS), maximized conversion value subject to a bound on the cost per value. Our results for \tcpa extend naturally to tROAS as well.}.

From an auction theoretic perspective however, this trend seems to go against foundational results that postulate that for profit-maximizers (\emph{aka} quasi-linear) bidders, it is optimal to use a classic bidding system (e.g., marginal CPA bidding, henceforth $\ua$) rather than using a $\tcpa$. In other words, for an advertiser with a quasi-linear utility functional, it is optimal to bid the marginal value in a truthful auction, while with $\tcpa$ bidding, an advertiser does not allow for a direct control over the marginal cost. So a natural question to ask is whether there is an intrinsic value for profit-maximizing bidders in using $\tcpa$ format? {Put simply, why do advertisers adopt $\tcpa$ bidding?}

One explanation given for the use of $\tcpa$-like bidding formats is that in many cases bidders may not be intrinsically profit-maximizers which is the classic assumption in economics, but instead care about high-level goal metrics, such as value maximization under return-of-investment constraints (\citet{AggarwalBM19,BalseiroDMMZ21}). Another related explanation is that the cognitive cost to bid and verify the outcome using $\tcpa$-like formats is low compared to other methods like fine-grained bidding or $\ua$.

In this paper, instead of reconsidering the profit-maximization framework, 
we show that even within this canonical framework we can explain the adoption of formats like $\tcpa$ bidding under a certain model and setting. \emph{In other words, the paper rationalizes profit-maximizers bidders adopting $\tcpa$ mechanisms.} We show that so long as a bidder believes that the auctioneer lacks commitment to follow the rule of the declared auction (e.g., readjusts reserve prices after bidders submit their bids), then the bidder prefers to use the $\tcpa$ format over the classic $\ua$ format.
\medskip

\noindent{\bf Model:} In our model (Sec.~\ref{sec:model} for details), there is a set of queries (ad slots) to be sold, and there is a \emph{single-shot} game between the auctioneer (auctioneer) and the bidders (advertisers). The auctioneer declares that the auction is a per-query second-price auction with a declared reserve-price, and then the bidders choose a bidding format (whether $\tcpa$ or $\ua$) and corresponding bids. After this, the auctioneer could then potentially deviate from the declared auction by readjusting the reserve prices, at potentially a per-query, personalized-per-bidder level. Finally, the outcomes are revealed.\footnote{{We note that our model is oblivious to the particular auto-bidding algorithms that auto-bidder uses so long as the optimal outcomes are reached.} }
    
Prominent in our formulation is the notion of \emph{commitment}. In a model \emph{with commitment} we assume that there is some mechanism by which the bidders can trust that the auction declared by the auctioneer will indeed be the one that is implemented. This is the standard model and assumes that there is some form of auditing available. In contrast, in a model \emph{without commitment}, there is no such mechanism or guarantee, and both the auctioneer and the bidders strategize their actions.  
\medskip

\noindent{\bf Results:}
At a high level, our main result is that if a bidder believes the auctioneer does not have commitment, then it prefers to use the $\tcpa$ bidding format over the $\ua$ format. \footnote{{While our model is in a single-shot game, we note that even in the context of repeated interactions where a bidder can monitor the outcome of the auctions over time, our theory can still directionally explain why some bidders may still prefer a $\tcpa$-format. For example, when a bidder is risk-averse about the auctioneer breaking commitment (e.g., bidder expects the auctioneer to slightly change the auction rules which cannot be observable for the bidder); or for a less-sophisticated bidder to whom is too costly to monitor their outcomes.}}

We begin (in Sec.~\ref{sec:3}) with a basic setting which captures much of the intuition and results in an instructive manner. We consider a 
model in which there is only one buyer in the auction facing an exogenous price landscape for the queries, e.g., a per-query floor set by the publishers who own the corresponding ad slots. In this model, we are able to derive not only the main result that the bidder prefers $\tcpa$ bidding (Theorem~\ref{th:1}), but also sharp results about the revenue implications for the auctioneer. We prove that
there is a \emph{credible} equilibrium (as defined in Sec. \ref{sec:3}) in which the auctioneer declares a reserve price of 0, and sticks to it, while the bidder chooses the $\tcpa$ format.
Under a mild assumption, this equilibrium is also shown to be beneficial to the auctioneer in terms of revenue as well as efficiency (Theorem~\ref{prop:rev_comp}).
        
We further study the case where the auctioneer has commitment to a reserve price, but the bidder mistakenly believes it doesn't. We show that, in this case, there is an instance where the revenue loss can be arbitrarily large (Prop.~\ref{prop:pr}). 


We continue to the general model with multiple bidders in Sec.~\ref{sec:gen}.
While we are able to solve for the equilibrium in the model of Sec.~\ref{sec:3}, solving the $\tcpa$ equilibrium in the general model with multiple bidders is hard,\footnote{\cite{multichannel} show that it is PPAD-hard to find an equilibrium when there are multiple $\tcpa$ bidders.} and studying the auctioneer's revenue implications seems intractable. Nonetheless, we show that a profit-maximizing bidder who believes that the auctioneer does not have commitment prefers the $\tcpa$ format to $\ua$ format (Theorem~\ref{th:2}). Our result heavily relies on the auctioneer's flexibility to readjust reserve prices at a per-query and bidder level reserve. For instance, when the set of queries is too large so readjusting at the query level is too expensive, the auctioneer is constrained to set reserve prices uniformly across queries. For this situation, we provide a companion result showing that under some technical assumption on the symmetry of the game (defined in Sec.~\ref{sec:gen}), that again a bidder prefers the $\tcpa$ format over the $\ua$ format (Theorem~4).

\medskip

\noindent{\bf Intuition:} To understand the intuition of our main result, that a bidder prefers the $\tcpa$ format over the $\ua$ format, let us decompose the auctioneer's objective -- to maximize revenue -- into two elements: the volume of queries the auctioneer sells and the price at which the queries get sold. These two components translate into two economic forces that the auctioneer considers at the moment of readjusting the reserve prices: (i) to induce a high marginal bid from the bidder (so to increase the volume of queries sold) and (ii) to set a reserve-price on the queries as closed as possible to the bidder's marginal bid. 
    Observe that force (i) is aligned with bidder's incentive while force (ii) goes in the opposite direction in regard to the bidder's utility.
    
In the $\ua$ format, force (i) disappears as the marginal bid is chosen by the bidder. Thus, the auctioneer's incentive in the no-commitment game is simply to price at the marginal bid. 
    
In the $\tcpa$ format, as the marginal bid is not fixed but rather is chosen by the auctioneer (given the cost per acquisition constraint), force (i) does not disappear, and hence, it mitigates the second force effect on the bidder's utility. This makes the bidder's utility higher in the $\tcpa$ compared to the $\ua$ format\footnote{In the specific model of Sec. \ref{sec:3}, we show that the former force completely dominates the latter, and the incentives of the auctioneer and the bidder are completely aligned.}. 
    
\begin{remark}[On interpreting the results]
Our model is that of a stylistic single-shot interaction between the two parties. There are some marketplaces, including ad auctions, where there is a repeated interaction between the bidders and the auctioneer that alleviates the commitment problem that we study here.
For example, with transparent reporting and avenues for experimentation, a bidder can monitor and verify or audit the auctioneers actions and commitments over time even with the $\ua$ format.
Furthermore, in practice, there exist certification processes and third-party audits such as Sarbanes-Oxley (SOX) compliance. Such audits are another method for bidders to have confidence that real-word auction systems are behaving as  described. We do not consider the repeated setting in this paper, nor the considerations of audits as a commitment device.
\end{remark}

\subsection{Related Work}
With the importance of auto-bidding in the industry, the topic has become of increasing interest in the literature. In a series of papers, the problem of formulation of $\tcpa$-like auto-bidding formats has been introduced and studied in various aspects. \citet{AggarwalBM19} introduced an optimization framework for auto-bidding, provided optimal bidding algorithms, and also studied the price of anarchy, i.e., the loss in efficiency due to $\tcpa$-like bidding, which was further  explored in subsequent work \citep{DengMMZ21,lmp22,yuanfpa,balseironeurips,mehta22}. There has also been recent work on understanding the optimal mechanism design in Bayesian settings~\cite{GolrezaiLP21,BalseiroDMMZ21} for ROI constrained bidders. The above papers consider the advertisers as having non-standard ROI utility functions in the main; \citet{BalseiroDMMZ21} explicitly distinguish between utility functions which are volume maximizing (ROI) versus profit maximizing. These works do not consider the question of how to rationalize the use of the $\tcpa$ format, but rather study additional consequences and auction design given an exogenous choice of the format. We can interpret our work as endogenizing the choice of format within the larger context between auction and bidding.

We mention that in a sense $\tcpa$ auto-bidding generalizes the more well-known budget constrained or financially-constrained model. This has been a very well studied topic with several lines of work studying this model, both from a fundamental perspective, and in the context of ad auctions, e.g.,~\cite{Laffont96optimal,Che2000TheOM, FeldmanMPS07,BalseiroG19, gaitondeetal23, tardosarxiv} among many others.

From an auction design perspective, our paper relates to three streams of the literature. The closest papers to our work consider a mechanism design problem with imperfect commitment where the auctioneer (designer) can readjust the rules after observing the agent's report. \cite{str1,str2}. Our paper differs from those previous results in that we study a multi-item problem (multiple queries) and focus on two main auction rules rather than a general mechanism design approach. 

Related to the first stream, \citet{akb, dask, wein} study the auction rules that are credible. That is, from the auctioneer's perspective, it is optimal to follow the rules of the declared auction. We contribute to this field by showing that, under some conditions, the $\tcpa$-mechanism is a credible mechanism for the multi-item cases (see Prop.~\ref{prop:6}).\footnote{\citet{dask} also construct a credible auction mechanism with good revenue guarantees.}

The third stream of papers study dynamic auctions where the auctioneer cannot commit to the auction that she will choose in future periods~\cite{gul,fud, skr1, skr2, liu}. These papers show that without commitment the auctioneer cannot obtain a revenue larger than the revenue from the efficient auction. While our work does not consider repeated auctions, our results also show that, under some circumstances, the lack of commitment limits the auctioneer revenue to the efficient outcome (see Prop.~\ref{prop:6}). 
 \section{Model}\label{sec:model}

Our model considers a platform (henceforth, the auctioneer) using a second price auction to sell $x\in X$ queries, where $X$ belongs to a measurable space $(X,\mathcal{A},\mu)$. Because our goal is to study the buyers' behavior under different auto-bidding mechanisms, we fix one of the buyers, henceforth the bidder, and assume that she has a private valuation $v\sim F$ for each query.

From the bidder's perspective, the cost per conversion on query $x$, $p(x)$, has two components: the {\em intrinsic} value of the query, which is given by $p_0(x)$ and the reserve price $r$ chosen by the auctioneer.\footnote{We assume that the intrinsic value $p_0:X\to \R_+$ is measurable function.} This intrinsic price function is quite general to include the case where the pricing comes from other buyers participating in the auction as well as pricing constraints set by the publisher {(see Sec.~\ref{sec:gen} for more details)}.\footnote{Our main result also includes the case where $p_0(x)$ may be unknown to the bidder.} Therefore, the bidder's payoff when she buys the queries with prices lower than $p$ is
$$ u(p|v) = \begin{cases}
0 &\mbox{ if } p< r \\
(v-r) H(r) + \int_{r}^p (v-z)dH(z) &\mbox{ if } p\geq r
\end{cases},
$$
where $H(p)= \mu(\{p_0(x)\leq p\})$ is the volume of queries with intrinsic price less than $p$.

We impose the following regularity condition on $H$.

\begin{assumption}\label{ass:0}
The price distribution $H$ is an integrable function (i.e., $\int_0^\infty H(z)dz<\infty$) and  is continuously differentiable with $h(p)=H'(p)>0$ (positive {\em density}).
\end{assumption}



\subsection{Auto-bidding mechanisms}

In order to bid for the queries $X$, the bidder has to choose among the different auto-bidding mechanisms the auctioneer offers to her. We consider two representative and commonly-used mechanisms. 

{\bf Marginal CPA}: The bidder submits a bid $b\in \R_+$ and the auto-bidding system bids $b$, on her behalf, on each of the queries $x\in X$. Thus, when the bid $b>\max\{r,p_0(x)\}$ the bidder receives query $x$ for a price $\max\{r,p_0(x)\}$. We conclude that the bidder's payoff is 
$$u(b|v; \ua) = \begin{cases}
    0 &\mbox{ if } b<r\\
    (v-r)H(r) + \int_r^b(v-z) h(z)dz &\mbox{ if } b\ge r
    \end{cases}.$$

{\bf TCPA:} The bidder submits a target cost per acquisition $T\geq 0$, and the auto-bidding system bids $b(T)$ in each of the queries $x\in X$ so that it maximizes the volume of queries subject to the average cost being no more than $T$. Thus, for $T>r$ we have that \begin{align}\label{eq:0} b(T;r) \in \arg \max \Big\{& H(b(T;r))|\mbox{ s.t. } r H(r) +\\
& \int_r^{b(T;r)} z h(z)dz \leq T \cdot H(b(T;r))\Big\}. \nonumber
\end{align}
The following lemma, whose proof is deferred to the appendix, shows that the above problem has a unique solution and the $\tcpa$ constraint is always binding.
\begin{lemma}\label{lem:1}
There is a unique solution $b(T;r)$ to Problem~\eqref{eq:0} which is the unique solution to equation \begin{equation}\label{eq:1} H(b)(b  - T) = \int_r^{b} H(z)dz.
\end{equation} Furthermore, for $T\geq r$, $b(T;r)$ is greater or equal than $T$, increasing in $T$ and decreasing in $r$.
\end{lemma}

From Lemma~\ref{lem:1}, we observe that given a target $T\geq r$, the bidder pays $T H(b(T;r))$ for the $H(b(T;r))$ queries she receives. We conclude that the bidder's payoff is 
    $$ u(T|v;\tcpa) = \begin{cases}
    0 &\mbox{ if } T<r\\
  (v-T)H(b(T;r)) &\mbox{ if } T\ge r
    \end{cases}.$$

\subsection{Credibility and commitment}

To model the auctioneer's credibility we consider the following four-stage game. 
\begin{enumerate}
    \item[(S1)] {\em Announcement:} the auctioneer announces a reserve price which applies for all queries $x\in X$.
    \item[(S2)] {\em Bidding:} the bidder chooses an auto-bidding mechanism and submits the bid to the auto-bidder, either marginal or target, accordingly.
    \item[(S3)] {\em Credibility:} the auctioneer potentially readjusts the reserve price, at a per-query and advertiser level.\footnote{{We restrict the auctioneer to use {\em reasonable} mechanisms where it can only modify reserve prices or bids (from competitors), similar to the model studied in \cite{akb}.}}$^,$\footnote{For the one bidder model of Sec.~\ref{sec:3}, it suffices to restrict to uniform readjustment of the reserve across all queries.} 
    \item[(S4)] {\em Auction is realized:} The auto-bidding system makes the per-query bids, and the final allocations and respective payments accrue.
\end{enumerate}

The third stage of the game is the key element for our analysis. We say that the auctioneer is credible when reserve prices do not change at S3. This could either be because the auctioneer has {\em commitment}, which means that the auctioneer commits not to change reserve prices (i.e. the game does not have S3); or because it has {\em endogenous credibility}, which means that along the equilibrium path it is optimal for the auctioneer not to change reserve prices. To distinguish between these two reasons, we denominate as the {\bf commitment game} the game without S3 and the {\bf no commitment game} as the game including S3.

The solution concept used in this paper is perfect Bayesian equilibrium.\footnote{Observe that for the auctioneer, once the bidder submits the bid, any belief about the bidder's type is irrelevant.}

 \section{One bidder model}\label{sec:3}

The purpose of this section is to distill, in an instructive manner, the main insight of our work by studying the simplest case: the bidder is the only buyer interested in the queries.\footnote{Thus, the intrinsic price $p_0(x)$ does not come from other bids, but instead it corresponds to particular constraints that publishers (owners of the queries) set on their queries.} The tractability of this model, which contrasts with the general model, also allows us to study the revenue consequences of lack of commitment on the auctioneer.

In this model, the auctioneer's revenue only comes from the queries sold to the bidder.\footnote{More generally, the auctioneer's revenue is a share of the transaction (the other fraction goes to the respective publisher). Our results easily extend to this general setting.} Therefore, given a bid of the bidder in either bidding format, the auctioneer's revenue is
\begin{align*}
    \pi (r|b;\ua) &= \begin{cases}
    0 & \mbox{ if } b<r \\
    r H(r) + \int_r^b z h(z)dz &\mbox{ if } b \geq r
    \end{cases}, \\
    \pi (r|T;\tcpa) &= \begin{cases}
    0 & \qquad\qquad  \mbox{ if } T<r \\
    T H(b(T;r)) & \qquad  \qquad\mbox{ if } T \geq r
    \end{cases}.
\end{align*}

\subsection{Commitment Game}

We first study the model where the auctioneer commits not to change the reserve prices after observing the bid from the bidder. In this situation, the auctioneer's problem resembles the classic optimal auction studied in \citet{myerson1981optimal}. 

Following the expository spirit of this section, we further simplify the analysis by imposing a standard regularity condition on the distribution $F$.

\begin{assumption}
The virtual valuation $\phi_F(v)= v - \frac{1-F(v)}{f(v)}$ is increasing.
\end{assumption}

\begin{proposition}\label{prop:1}
For the commitment game, the revenue-maximizing mechanism is $(\chi^*,\tau^*)$ 
\begin{align*}
 \chi^*(v)= \begin{cases}
0 &\mbox{ if } v < r_{\mye} \\
v &\mbox{ if } v\geq r_{\mye}
\end{cases} \qquad 
\tau^*(v) = r_{\mye} H(r_{\mye}) +\int_{r_{\mye}}^v z h(z)dz
\end{align*}
where $r_{\mye}= \phi_F^{-1}(0)$. Here, the allocation function $\chi(v)$ and transfer function $\tau(v)$ mean that the bidder receives queries $\{x: p_0(x)\leq \chi(v)\}$ for a payment of $\tau(v)$. 
\end{proposition}

Proposition~\ref{prop:1} characterizes the revenue-maximizing policy among all feasible auction and bidding formats (in particular, including the $\ua$ and $\tcpa$ formats). We defer the proof to Appendix~\ref{app:simple}. 

The next proposition shows that the auctioneer can implement this optimal mechanism in both formats, and hence, making them equivalent. 

\begin{proposition}\label{prop:2}
In the commitment model, the $\ua$ and $\tcpa$ mechanisms are equivalent along the equilibrium path. More precisely, in any equilibrium, the auctioneer sets a reserve price $r_{\mye}$; the bidder either bids $b=v$ on the $\ua$ mechanism or $T$ in $\tcpa$ mechanism such that $b(T;r_{\mye})= v$. 

\noindent In particular, we obtain that in the commitment game:
\begin{itemize}
    \item[(i)] the auctioneer's revenue: \\$\pi_C^*= \E_v[\mathbf{1}_{\{v\ge r_{\mye}\} } \left(vH(v)-\int_{r_{\mye}}^v H(z)\right)dz]$,
    \item[(ii)] the bidder's utility: $u_{C}^* = \E_v[\mathbf{1}_{\{v\ge r_{\mye}\} } \int_{r_{\mye}}^v H(z)dz]$,
    \item[(iii)] the welfare: $W_C= \E_v[\mathbf{1}_{\{v\ge r_{\mye}\} }vH(v)]$.
    \end{itemize}
\end{proposition}

\begin{proof}
First, notice that since the auctioneer does not change the reserve price in the commitment game, the optimal bidding strategy for the bidder is to submit a marginal bid equal to her value. She can do this either directly using the $\ua$-format or indirectly, in the $\tcpa$ format by submitting a target $T$ that induces the same marginal bid.
Thus, from the bidder's perspective, the $\ua$ format and the $\tcpa$ format are equivalent.
Hence, from the auctioneer's perspective he can implement the optimal mechanism of Proposition~\ref{prop:1} by announcing a reserve price of $r_{\mye}$ at S1. 
From Proposition~\ref{prop:1} we have that the bidder pays $T^*(v)= \mathbf{1}_{\{v\ge r_{\mye}\} } \left( vH(v)- \int_{r_{\mye}}^v H(z)\right)$. The remaining claims of the proposition are direct consequence of this characterization.
\end{proof}

\subsection{No-commitment game}

This section shows that when the auctioneer lacks commitment, the two auto-bidding mechanisms are no longer equivalent. We show that if the bidder chooses the $\ua$ format, the final auction turns to be equivalent to a first-price auction (FPA). By contrast, when the bidder opts for the $\tcpa$ format, the final auction turns to be equivalent to a second-price auction without reserve (SPA). \medskip 

{\noindent \bf The $\ua$ Subgame}\medskip

We first characterize the set of equilibria for the subgame where the bidder chooses the $\ua$ format. We present the following straightforward lemma and a direct consequence of the lemma in Proposition~\ref{prop:3}.

\begin{lemma}\label{lem:2}
Consider the subgame where the bidder chooses the $\ua$ format and bids $b$. Then, if the reserve price $r$ announced at S1 is such that $r\neq b$, then the optimal decision for the auctioneer is to readjust the reserve price to $r=b$. 
\end{lemma}

\begin{proposition}[$\ua$ equivalent to FPA]\label{prop:3}
The auction where the bidder chooses the $\ua$ format is equivalent to a FPA. Thus, for each valuation type $v$ the bidder submits a bid $b^*(v)$ solving
\begin{equation}\label{eq:fpa} 
\max_{b} (v-b)H(b).
\end{equation}
\end{proposition}

The previous result describes the negative effect of the lack of commitment from the auctioneer. Under the $\ua$ format, the auctioneer learns how much the bidder is willing to pay for each query; therefore, the auctioneer's sequential rationality pushes him to charge such value to the bidder. Thus, the auctioneer cannot credibly commit to keeping any reserve price announced at S0. Anticipating this effect, the bidder shades the bid and by consequence, she gets fewer queries allocated to her compared to the commitment case. 

We formalize this discussion in the following proposition.

\begin{definition}
We say that an equilibrium is credible if the auctioneer does not change the reserve price after observing the bid.\footnote{{Our definition is similar to the self-enforcement agreement theory studied in \citet{aumann1990nash}.}}
\end{definition}

\begin{proposition}\label{prop:ua}
There does not exist a credible equilibrium such that the bidder chooses the $\ua$ format. Furthermore, all equilibria are payoff equivalent inducing expected payoffs $\pi^*({\ua}) = \E_v[b^*(v) H(b^*(v)) ]$, $u^*({\ua}) =\E_v[(v-b^*(v)) H(b^*(v)) ]$ and expected welfare $W_{\ua}= \E_v[v H(b^*(v)) ] $, where $b^*$ is the solution to Problem~\eqref{eq:fpa}. 
\end{proposition}

\begin{proof}
From Proposition~\ref{prop:2} we have that in any equilibrium, the bidder with valuation-type $v$ bids $b^*(v)$ and the auctioneer sets a final reserve price so that $r=b^*(v)$. From the envelope theorem we see from Problem~\eqref{eq:fpa} that ${b^*}'(v)  = H(b^*(v))$ which implies that $b^*(v)$ is increasing on $v$. This implies that the auctioneer's initial reserve price $r$ has to be different from at least one of $b^*(v)$ and $b^*(v')$ when $v\neq v'$. This implies that in, any equilibrium, the auctioneer readjusts the reserve price for at least one such bid. We conclude that the there is not a credible equilibrium in the $\ua$ subgame. 

The payoffs described in the proposition are a consequence of that, in any equilibrium, the bidder bids $b^*(v)$ gets all queries that have intrinsic prices less than $b^*(v)$ for a price $b^*(v)$.
\end{proof}

\medskip

\noindent{\bf The $\tcpa$ Subgame}\medskip

Similar to the previous analysis, we characterize the set of equilibria for the subgame where the bidder chooses the $\tcpa$ format.

\begin{lemma}\label{lem:4}
Consider the subgame where the bidder chooses the $\tcpa$ format and bids $T$. Then, if the reserve price $r$ announced at S1 is such that $r>0$, then the optimal decision for the auctioneer is to readjust the reserve price to $r=0$.
\end{lemma}

\begin{proof}
Clearly, if the auctioneer would a set a reserve price $r\leq T$ otherwise he gets zero profits. For reserve price $r\leq T$, observe that the auctioneer's revenue is $T H(b(T;r))$. From Lemma~\ref{lem:1} we have that $b(T;r)$ is decreasing in $r$, and hence, since $H$ is increasing the optimal reserve price is $r=0$.
\end{proof}

The intuition behind this lemma is that the maximum price to pay is not predetermined as in $\ua$ case. Instead it is chosen by the auctioneer to meet the $\tcpa$ constraint. Because the auctioneer's revenue is proportional to the volume of queries allocated to the bidder, to maximize such volume, it is optimal to set reserve price $r=0$. Therefore, the bidder by letting the auctioneer bid on her behalf, makes the auctioneer internalize the negative effect of rent extraction via a reserve, since the latter leads to the auto-bidder decreasing the final marginal bid $b(T;r)$. 

\begin{proposition}[$\tcpa$ equivalent to SPA]\label{prop:tcpa}
The auction where the bidder chooses the $\tcpa$ format is equivalent to a SPA. For each valuation type $v$ the bidder submits a target $T^*(v)$ such that $b(T^*(v);0)= v$. That is, $T^*(v)$ solves
\begin{equation}\label{eq:tcpa}
H(v)(v- T^*(v))=\int_0^v H(z)dz.
\end{equation}
\end{proposition}

\begin{proof}
Consider any equilibrium of the game. Because in stage S3, the best response of the auctioneer is to set a reserve price $r=0$ (Lemma~\ref{lem:4}), the bidder's optimal response at S2 consists in submitting a target so that the marginal bid equals to her valuation $v$ for a price landscape without reserve price. Hence, she submits a target $T^*$ as described in Equation~\eqref{eq:tcpa}.
\end{proof}

This proposition starkly contrasts with Proposition~\ref{prop:ua}. In the $\tcpa$ format, the bidder believes that the auctioneer will reduce the reserve price while, in the $\ua$ case, the bidder believes that the auctioneer will increase the reserve price to its bid. This is because in the $\tcpa$ format, once the bidder bids the target $T$, the auctioneer's incentives are fully aligned with the bidder's incentive:  the auctioneer only cares about maximizing the volume of sold queries. 

A second difference with $\ua$ format is that, in this case, the auctioneer sets a reserve price independently of the bidder's target. In particular, if the auctioneer announces at S0 a reserve price $r=0$, he can credibly commit to that price. 

The following proposition summarizes these findings.

\begin{proposition}\label{prop:6}
A credible equilibrium exists when the bidder chooses the $\tcpa$ format. The seller sets a initial reserve price $r=0$, the bidder bids $T^*$ (the solution to Equation~\eqref{eq:tcpa}). Moreover, every equilibrium is payoff equivalent inducing expected payoffs $\pi^*({\tcpa}) = \E_v[T^*(v) H(v) ]$, $u^*({\tcpa}) =\E_v[(v-T^*(v)) H(v) ]$ and expected welfare $W_{\tcpa}= \E_v[vH(v) ] $.
\end{proposition}

\begin{proof}
The following is a credible equilibrium of the game: the auctioneer sets a reserve price $r=0$ at S1, the bidder chooses the $\tcpa$ format and submits the target $T^*$ which solves Equation~\ref{eq:tcpa} at stage S2, and the auctioneer keeps the reserve price $r=0$ at stage S3.

Furthermore, the subgame after S1 is the same for any initial reserve price that auctioneer announces at S1 (see Proposition~\ref{prop:tcpa}). Therefore, all equilibria are payoff equivalent. And the expected payoff are an immediate consequence from the fact for every type $v$, the target is $T^*(v)$ is binding in Problem~\eqref{eq:0} and hence she gets queries $\{x: p_0(x)\leq v\}$ for price $T^*(v)H(v)$. 
\end{proof}

 An important consequence of the above results is that without commitment, the auctioneer allocates the queries efficiently (i.e., maximizes welfare).\footnote{This result  is reminiscent of the Coasean literature, which, in a different context (durable good monopolist),  studies conditions in which a monopolist does not exercise his monopolistic power and allocates efficiently (see \citet{bulow82}).} 

\begin{corollary}\label{coro:tcpa}
In any equilibrium of $\tcpa$ subgame, the auctioneer efficiently allocates the queries.
\end{corollary}

\subsection{Utility, Welfare, and Revenue implications in the No-commitment game}
We start this section by showing our main result for the specific setting of Section~\ref{sec:3}: when the auctioneer lacks commitment, it is optimal for the bidder to bid according to the $\tcpa$ format. Thus, our result provides a rational explanation as to why quasilinear bidders opt for $\tcpa$ auto-bidding mechanisms.

\begin{theorem}\label{th:1}
In any equilibrium, for every type-valuation $v$, the bidder strictly prefer to choose the $\tcpa$ format over the $\ua$ format. Thus, $\pi^*_{NC}= \pi^*(\tcpa)$, $u^*_{NC}=u^*(\tcpa)$ and $W_{NC}=\E_{v}[vH(v)]$.  
\end{theorem}

\begin{proof}
Consider $b^*(v)$ defined in Equation~\eqref{eq:fpa}. Then,
\begin{align*}
 u^*(\ua|v)  &= (v-b^*(v)) H(b^*(v))\\
    &= (v-b^*(v)) H(b^*(b^*(v);b^*(v)))\\
    &< (v-b^*(v)) H(b(b^*(v);0))
    \\&=u(b^*(v)|v;\tcpa)
    \\&\le u^*(\tcpa|v).
\end{align*}
The first equality is by definition of $b^*$ (Proposition~\ref{prop:ua}). The second equality holds because $H$ is increasing and, hence, with $\tcpa$ constraint of $b^*(v)$ and reserve price $b^*(v)$ the optimal bid is to buy all queries with price less or equal than $b^*(v)$ (i.e. $H(b^*(v))$). The first inequality holds due to Lemma~\ref{lem:1}. The next equality is by definition of bidding a target $T=b^*(v)$. The last inequality is by definition of $u^*(\tcpa)$, the bidder's payoff using the $\tcpa$ format with the optimal target. 
\end{proof}

Regarding the welfare implications, Corollary~\ref{coro:tcpa} shows that in the $\tcpa$ format, the final allocation is is welfare-optimal. On the other hand, because in the $\ua$ case the bidder shades her bid (i.e., $b^*(v)<v$), we have that the allocation is inefficient. We conclude that, from a welfare perspective, the $\tcpa$ format is preferable compared to the $\ua$ format.\medskip

\noindent{\bf Revenue implications}\medskip

Our previous results show that from the bidder's (and also welfare) perspective, when the auctioneer does not commit, the bidder prefers to use a $\tcpa$ mechanism over the $\ua$ mechanism. We now tackle the question from the auctioneer's angle. Does having the $\tcpa$ format cause a loss in revenue for the auctioneer?
The following result shows that, under some reasonable assumption on $H$, the auctioneer itself benefits from offering a $\tcpa$ format to the bidder.  

\begin{assumption}\label{ass:2}
$H$ satisfies that $vh(v)$ is non-decreasing.
\end{assumption}

Assumption~\ref{ass:2} implies that the marginal revenue to sell the queries at price $p$, $p h(p)$, is non-decreasing on $p$. This assumption is quite natural and holds in common settings, for instance, when $H$ is convex, in which case $h$ is non-decreasing. 

\begin{theorem}[Revenue Comparison]\label{prop:rev_comp}
 Suppose that Assumption~\ref{ass:2} holds. Then for every valuation-type $v$ we have that $\pi^*({\tcpa}|v)> \pi^*({\ua}|v)$. Furthermore, for every $\gamma>0$, we can find an instance $\langle F,H \rangle$ such that $\pi^*({\tcpa})>\gamma \cdot  \pi^*(\ua)$.
\end{theorem}

\begin{proof}
Consider the auxiliary function $w(v) := \int_0^v zh(z)dz$. Using integration by parts we have that $w(v)=vH(v)-\int_0^v H(z)dz$, and using the defintion of $T^*(v)$ in Equation~\eqref{eq:tcpa}, we get that $w(v)=T^*(v)H(v)$. 

Assumption~\ref{ass:2} implies that $w$ is a convex function. Thus, for every $v,v'$, we have that
\begin{equation}\label{eq:prs1} 
w(v)\ge w(v') + \underbrace{w'(v')}_{v'h(v')} (v-v').
\end{equation}

Take $v'=b^*(v)$, where $b^*(v)$ is the optimal bidding in the $\ua$ format. Taking the first order conditions on Problem~\eqref{eq:fpa} (the solution has to be interior in this problem), we have $h(b^*(v))(v-b^*(v)) = H(b^*(v))$. Therefore, replacing $v'$ in Equation~\eqref{eq:prs1} we obtain that 
$$ T^*(v)H(v) \ge w(b^*(v)) + b^*(v) H(b^*(v)).$$

Because $h>0$, we have that $w>0$. Therefore, we obtain that $T^*(v)H(v) > b^*(v) H(b^*(v))$, or equivalently, $\pi^*({\tcpa}|v)> \pi^*({\ua}|v)$.

To finish the proof, consider $F(v)=v$ with support in $[0,1]$ (uniform distribution) and $H_n(p) = p^n$ for $p\in [0,1]$. Simple computations shows that for every $v\in [0,1]$, $b^*(v) = \frac v {n+1}$, and hence, 
$$\pi^*({\ua}) =\int_0^1 \left(\frac v {n+1} \right)^n dv \leq \frac 1 {(n+1)^n}.$$
On the other hand, the solution to Equation~\eqref{eq:tcpa} is $T^*(v) = \frac{n}{n+1}v$. This implies that 
$$ \pi^*({\tcpa}) = \int_0^1 \frac{n}{n+1}v\cdot  v^{n} dv =\frac n {(n+1)(n+2)}.$$
Thus,
$$ \frac{\pi^*({\tcpa})}{\pi^*({\ua})} \geq \frac{\frac n {(n+1)(n+2)}} {\frac 1 {(n+1)^n}} = \frac {n (n+1)^{n-1}}{n+2}.$$
We conclude by taking $n$ large enough such that $\frac {n (n+1)^{n-1}}{n+2}>\gamma$.
\end{proof}

Theorem~\ref{prop:rev_comp} shows that the revenue on the SPA (the $\tcpa$ format) is not equivalent to the revenue obtained in the FPA (the $\ua$ format). At first glance, this result seems to contradict the well-known Revenue Equivalence Theorem (RET) between these auctions (see Chapter 1.3 of \citet{krishna2010} for a textbook treatment). While the result is true when the auctioneer only owns one query (in both auctions the revenue is simply $p_0(x)$), having uniform bidding among heterogeneous queries constraints the bidder to shade the bid so that she receives a suboptimal fraction of queries. Thus, the fraction of queries allocated in the SPA is larger than in the FPA, violating the main condition for the RET\footnote{If the bidder had knowledge of the pricing $p_0(x)$ and could bid independently in each of the queries, the FPA would allocate exactly the same as in the
SPA, and therefore, RET would apply in such a case.}.

Even though the condition imposed in Assumption~\ref{ass:2} considers a wide range of cases, the following instance -- that does not satisfy Assumption~\ref{ass:2} -- provides an example where the auctioneer prefers not to offer the $\tcpa$ mechanism to the bidder.

\begin{proposition}\label{prop:ce}
Suppose that the valuation type have support on $[1,\infty)$ and consider the instance
$$ \hat H(p) =\begin{cases}
0 &\mbox{ if } p < 1 \\
1-\frac 1 p &\mbox{ if } p\ge 1 
\end{cases}.
$$
Then, $\pi^*({\tcpa}|v)= \log(v)$ and $\pi^*({\ua}|v) = \sqrt v - 1$ which implies that $\pi^*({\ua}|v)>\pi^*({\tcpa}|v)$. In particular, for every $\gamma>0$ we can find a distribution $F$ so that the instance $\langle F,\hat  H \rangle $ is such that $\pi^*({\ua})>\gamma \cdot\pi^*({\tcpa})$.
\end{proposition}

\subsection{The value of commitment}
We finish this section by measuring the value to auctioneer of having commitment in the game. First, we measure the revenue loss of the auctioneer when the auctioneer does not have a commitment mechanism to keep the announced reserve prices, compared to the commitment benchmark. The second measure studies the revenue loss of the auctioneer when, even though he can commit not to change the reserve prices, the bidder has a mistaken belief that the auctioneer would actually readjust reserve prices\footnote{The bidder's decision on how to bid and which auto-bidding system to choose purely depends on her belief about the auctioneer's commitment. Thus, Theorem~\ref{th:1} does not rely on what the auctioneer is actually doing in terms of readjusting reserve prices, but instead, it relies on what the bidder believes the auctioneer is doing (see also, Remark~\ref{rem:bidder-belief}).}. All proofs of this section are delegated to Appendix~\ref{app:simple}.

\begin{proposition}[The value of commitment]\label{prop:vc}
Assume that the distribution $F$ has support in $[\underline v, \overline v]$. {Denote by $\psi =\frac{\pi_{\underline v\mbox{-}\spa}}{\pi_{\mye}}$, the relative revenue of selling one item with a SPA with reserve price $\underline v$ compared to sell it using revenue-optimal auction.} Then for every instance $\langle F, H\rangle$ we have that $\pi^*_{NC} \geq \psi \cdot \pi^*_C$. Moreover, the bound is tight: an instance $\langle F, H\rangle$ exists such that $\pi^*_{NC} = \psi \cdot \pi^*_C$
\end{proposition}

Proposition~\ref{prop:vc} shows the more homogeneous impressions are, the higher is the loss of revenue for the auctioneer when he does not have commitment. 

\begin{proposition}[The value of showing commitment]\label{prop:pr}
We denote by $\pi^*_{WB}$, the the auctioneer's revenue when the bidder bids the non commitment optimal reserve price $T^*$ (see Equation~\eqref{eq:tcpa}) but the auctioneer keeps the reserve price to $r=r_{\mye}$. That is, $\pi^*_{WB}= \pi^*(r_{\mye}|T^*;\tcpa)$. Then, for every $\gamma>0$ an instance $\langle F, H \rangle$ exists such that $\pi^*_{C} > \gamma \cdot \pi^*_{WB}$.
\end{proposition}

The last proposition shows the cost of having a misguided bidder in the auction. It further shows that if the auctioneer believes that bidder does not trust the auctioneer's commitment, then the rational choice for the auctioneer is to declare a reserve price of $0$ instead of $r_{\mye}$. The following remark re-emphasizes this point.

\begin{remark}
\label{rem:bidder-belief}
Our result (including the general version of Section~\ref{sec:gen}) that a bidder prefers the $\tcpa$-format only depends on the bidder's belief about whether the auctioneer can commit to the auction rules he declares. Therefore, even a committed auctioneer may benefit by offering $\tcpa$-like mechanisms as it increases the value for bidders who are skeptical about the auctioneer's commitment. 
\end{remark}

 \section{The general model}\label{sec:gen}

This section generalizes Theorem~\ref{th:1} to a setting where {\em the bidder} faces competition by other buyers (we will call them extra-buyers), that are interested in the queries $X$. 

We start the section by unfolding the intrinsic price $p_0$ when we have multiple buyers. In particular, we allow the intrinsic price to be a random variable whose uncertainty comes from the extra-buyers bidding strategies that may be private, the different conversion probabilities and the publishers' pricing constraint that may be unknown to the bidder. 

\subsection*{The intrinsic price with multiple buyers}

We consider $n$ extra-buyers (aside from our original bidder) participating in the auction. Each extra-buyer $i$ has a private valuation $v_i$ per conversion on the query, strategically chooses an auto-bidding format and submits a bid according to the format.
We denote by $\bm{\sigma}(\v) = (\sigma_1(v_1),\ldots,\sigma_n(v_n))$ the extra-buyers' strategies.  

Let $\b= (b_1,\ldots, b_n)$ be the final marginal bids of the extra-buyers; $\q(x)= (q_{0}(x), q_{j}(x))_{j=1}^n$ the probability of conversion on query $x\in X$ for the bidder and the buyers (respectively); and $p_B(x)$ the  pricing constraint set by the publisher owning query $x$.\footnote{The functions $\q, p_B$ are assumed to be integrable functions.} Then, in a second-price auction the realized intrinsic price the bidder faces for a conversion on query $x\in X$ is 
$$ p_0(x|\b,\q,p_B) =  \max\left\{ \frac{\max_{j=1,\ldots,n} b_j \cdot q_{j}(x)}{q_0(x)}, p_B(x)\right\}.$$
Hence, the realized price distribution is $$H(p| \b,\q,p_B) = \int_{\{x:\;  p_0(x|\b,\q,p_B)\leq p \}}q_{0}(x) d\mu(x).$$

From the bidder's perspective, we consider the following information structure. The bidder believes that extra-buyers valuations $\v = (v_i)_{i=1}^n$ are drawn independently, with distribution $F_i$ with support $[0,\overline v_i]$ for $i=1,\ldots,n$.\footnote{Because the valuations can be arbitrarily small, we have that for every valuation-type $v$ of the bidder, $\P [v q_0(x) > v_j q_j(x) \mbox{ for } x\in X \mbox{ (a.e.)} ]>0$. That is, with positive probability, the bidder by bidding her value is the strongest candidate in the auction.} The bidder also assesses that the conversion probabilities $\q$ and the pricing constraint $p_B$ are drawn according to $(\q, p_B)\sim G|_{v}$. We do not necessarily impose that extra-buyers are playing equilibrium strategies but instead impose that they are individually rational: an extra-buyer never bids above her valuation.
\begin{assumption}\label{as:2g}
For every Extra-Buyer $i$, we have that $\P[b_i \leq v_i \mbox{ for } i=1,\ldots,n] =1$, where $b_i$ is Extra-Buyer $i$'s final marginal bid.
\end{assumption}

Regarding the intrinsic price the bidder faces, observe that in the presence of extra-buyers using $\tcpa$ format, the final marginal of those $\tcpa$-extra buyers depends on the bidder's marginal bid $b_0$. Therefore, the intrinsic price the bidder faces is a random variable $p_0(\cdot|\omega_{v},b_0)$, where $\omega_v=(\bm{\sigma}(\v),\q,p_B)$ is the random variable containing the bidder's unknown terms. Consequently, the price distribution is a random variable $H(p|\omega_{v},b_0)$. This is precisely why the full model is harder to analyze than the model of Sec.~\ref{sec:3} -- the price distribution for the bidder is now a function of its own bid. 

We assume that $H$ satisfies the general version of Assumption~\ref{ass:0} for random variables.\footnote{All assumptions holding for $\omega_v$, are stated up to a zero measure set.}
\begin{assumption}\label{as:0g}
$H(p|\omega_v,b_0)$ satisfies Assumption~\ref{ass:0} for every $b_0$, $\omega_v$.
\end{assumption}
\subsection*{The value of the $\tcpa$ format}

After the previous prelude, we are now in a position to state our main result: in the non-commitment model, the bidder prefers to use the $\tcpa$ format. More precisely, we show that if the auctioneer can readjust the auction rule to set a per-query and personalized-per-bidder reserve prices, then the bidder's expected payoff using the $\tcpa$ format is larger than the expected payoff using the $\ua$ format. We also provide a companion result in the supplementary material showing that when the extra-buyers game is {\em symmetric} in a certain sense 
and the auctioneer is constrained to set a uniform reserve price across queries, then again, the bidder prefers to use the $\tcpa$ format.

Recall the notation that $u^*(\tcpa|v)$ denotes  the expected payoffs when the bidder with valuation $v$ chooses the $\tcpa$ format and submits an optimal target. Similarly, $u^*(\ua|v)$ is the expected payoff when the bidder chooses the $\ua$ format and submits an optimal marginal bid. As our main technical result, we show:
\begin{theorem}\label{th:2}
Suppose that the auctioneer can readjust reserve prices to per-query and personalized-per-bidder level. Then, $u^*(\tcpa|v) >  u^*(\ua|v)$.
\end{theorem}

\begin{proof}
Fix the valuation of the bidder to $v$ and let $b_{\tiny \ua}\in (0,v]$ be an arbitrary bid. 

Suppose that the bidder submits $b_{\tiny \ua}\in [0,v]$ using the $\ua$ format and assume that the auctioneer has chosen the optimal reserve prices (call this scenario ``world $M$"). We claim that the bidder weakly improves her payoff by bidding a target $T=b_{\tiny \ua}$ with the $\tcpa$ format (call this scenario ``world $T$") for every realization $\omega_v$. Furthermore, the inequality is strict for a positive measure of $\omega_v$. This claim proves the theorem, and we show the claim in the following four steps. \medskip

\noindent{\bf Step 1.} Let $X_{\tiny \ua}(
\omega_v)$ be the subset of queries that the bidder obtains in world M. Then $$u(b_{\tiny \ua}|\ua;v;\omega_v) = (v-b_{\tiny \ua})\int_{X_{\tiny \ua}(\omega_v)} q_0(x)d\mu(x),$$ where $\mu$ is the measure on the space of queries $X$. Indeed, by optimality of the auctioneer, the reserve price for the bidder on those queries must be $r=b_{\tiny \ua}$.

\noindent{\bf Step 2.1.} We prove that the revenue of the auctioneer in world $T$ 
is at least the revenue in world $M$. This is because one strategy for the auctioneer in world $T$ is to simply set a reserve price of $b_{\tiny \ua}$ for the bidder (equal to its target). Under this, the situation is identical to world $M$ for the bidder (due to Step 1) and for every buyer (because the auto-bidder for our bidder bids $b_{\tiny \ua}$ on all queries as the reserve is set to the target), yielding the same revenue as in world $M$\footnote{In case of multiplicity of equilibria, we assume that the same bidding equilibrium arises on worlds $M$ and $T$ as they are indistinguishable to the agents.}. 

\noindent{\bf Step 2.2.} We next prove that the revenue that the auctioneer obtains \emph{from the bidder} in world $T$ is greater or equal than the revenue from the bidder in world $M$. Suppose for the sake of a contradiction that this is not true.
Due to Step 2.1, this means that the revenue obtained from the extra bidders is higher in world $T$ than in world $M$.
We now leverage the fact that the auctioneer can set a reserve price at query/bidder level, in order to recreate the situation from world $T$ in world $M$. For each extra-buyer, the auctioneer can add a per-query personalized reserve equal to the 
bid the bidder submits for the query in world $T$. 
Moreover, for the queries $X_{\tiny \tcpa}(\omega_v)$ that the bidder wins in world $T$, the auctioneer can set a high reserve price on the extra-buyers so that the only feasible candidate is the bidder. 
With this simulation, the auctioneer can recreate the extra-buyers' bidding behavior from world $T$ in world $M$. Thus the auctioneer changed reserves to obtain the same revenue from the extra-buyers in world $M$ as in world $T$. In this way (only by changing reserves) one could increase the revenue in world $M$. This contradicts the optimality of the reserve prices the auctioneer chooses in world $M$.



\noindent{\bf Step 3.} It follows from Step 2.2 that the volume of conversions obtained by the bidder is higher in world $T$ than in world $M$. This is simply because the revenue from the bidder equals the average cost per conversion times the volume of conversions. The average cost-per-conversion in both the worlds is the same, equal to $b_{\tiny \ua}$ -- in world $T$ because we assume the target is binding (Assumption~\ref{as:0g} and Lemma~\ref{lem:1}), and in world $M$ because the reserve is set to the same value. Since the revenue from the bidder is higher in world $T$, we conclude that the volume of conversions is higher in world $T$.

\noindent{\bf Step 4.} We assert that $u(b_{\tiny \ua}|\tcpa;v;\omega_v)\geq u(b_{\tiny \ua}|\ua;v;\omega_v)$. Indeed, observe that
\begin{align}
    u(b_{\tiny \ua}|\tcpa;v;\omega_v) &= (v-b_{\tiny \ua})\int_{X_{\tiny \tcpa}(\omega_v)} q_0(x)d\mu(x) \nonumber\\
    &\geq (v-b_{\tiny \ua})\int_{X_{\tiny \ua}(\omega_v)}q_0(x)d\mu_(x) \nonumber
    \\&= u(b_{\tiny \ua}|\ua;v;\omega_v) \nonumber
\end{align}
where the first equality holds because the target of the bidder in world $T$ is set to $b_{\tiny \ua}$, and the $\tcpa$ constraint is binding (from Assumption~\ref{as:0g} and Lemma~\ref{lem:1}). The first inequality is from Step 3.  The final equality is from Step 1.

\noindent{\bf Step 5.} In step 4, we already proved the weak inequality from the theorem; now we prove the strict inequality. We prove that there exists a positive measure of events $\omega_v$ such that $u(b_{\tiny \ua}|\tcpa;v;\omega_v)> u(b_{\tiny \ua}|\ua;v;\omega_v)$. Indeed, because 
$b_{\tiny \ua}>0$, a positive measure of events $\omega_v$ exists such that the extra-buyers' valuations are small enough so that for every query $x\in X$, $\max_{i=1,\ldots,n} v_i q_i(x) < b_{\tiny \ua} q_0$. Hence, by Assumption~\ref{as:2g} the auctioneer never allocates queries to extra-buyers. 
Thus, in world $M$, the auctioneer's revenue is $b_{\tiny \ua} \cdot H(b_{\tiny \ua}|\omega_v)$. In world $T$, if the bidder bids a target $T=b_{\tiny \ua}$, then by reducing the bidder's reserve price to $r<b_{\tiny \ua}$ we have the final marginal bid is $b(b_{\tiny \ua};r)>b_{\tiny \ua}$. Thus, the auctioneer obtains a revenue of $T\cdot H(b(b_{\tiny \ua};r)|\omega_v) = b_{\tiny \ua} \cdot H(b(b_{\tiny \ua};r)|\omega_v)$. This is strictly larger than the revenue in world $M$ due to Assumption~\ref{as:0g}. We conclude that auctioneer sets a reserve price $r<b_{\tiny \ua}$ in world $T$. Thus, we obtain
\begin{align*} 
u(b_{\tiny \ua}|\tcpa;v;\omega_v) &= (v- b_{\tiny \ua}) H(b(b_{\tiny \ua};r)|\omega_v) \\
 &> (v- b_{\tiny \ua}) H(b_{\tiny \ua}) = u(b_{\tiny \ua}|\ua;v;\omega_v).
\end{align*}

\medskip

To conclude the proof of the theorem, from Step 4. and Step 5. we derive that
\begin{align*}
u(b_{\tiny \ua}|\ua;v) &= \E_{\omega_v}[u(b_{\tiny \ua}|\ua;v;\omega_v)] \\
&< \E_{\omega_v}[u(b_{\tiny \ua}|\tcpa;v;\omega_v)]  \leq u^*(\tcpa|v).
\end{align*}
Since $b_{\tiny \ua}$ is arbitrary, we get that $ u^*(\tcpa|v)< u^*(\ua|v)$.\end{proof}

Theorem~\ref{th:2} strongly relies on the auctioneer's ability to readjust the reserve price at the query and bidder level. However, when the set of queries is large, readjusting the reserve prices for each query may turn out to be too expensive. In this kind of situation, when the auctioneer is constrained to set a uniform reserve price, we provide a companion result showing that the bidder still prefers the $\tcpa$ format under some symmetry condition on the extra-buyers game (see the supplemental material for details).

\subsection{Uniform reserves}
The proof for Theorem~\ref{th:2} strongly relies on the auctioneer's ability to readjust the reserve price at the query and bidder level. However, when the set of queries is large, readjusting the reserve prices for each query may turn out to be too expensive for the auctioneer. In this kind of situation, when the auctioneer is constrained to set a uniform reserve price, Theorem~\ref{th:3} shows that the bidder still prefers the $\tcpa$ format so long as the extra-buyers game is symmetric. We remark that this result does not rely on how the auctioneer readjusts the extra buyers' reserve prices.\footnote{This weaker condition on how the auctioneer behaves with the extra-buyers allows to include cases like when some extra-buyers are budgeted constrained, and hence, the auctioneer cannot readjust their reserve prices.}

When dealing with uniform reserve prices, the key technical challenge compared to Theorem~\ref{th:2} is that the auctioneer cannot replicate the effect of the bidder's bidding on the remaining extra-buyers by setting personalized uniform reserve prices. Thus, when the auctioneer readjusts the bidder's reserve price not only the bidder's marginal bid changes but also the marginal bids of extra-buyers using a $\tcpa$ format. To tackle this problem, we assume that game for extra-buyers using the $\tcpa$ format is symmetric.

\begin{definition}\label{def:sym}
The extra-buyers' game is $\tcpa$-symmetric if for every $\omega_v$ and Extra-Buyers $i,j$ using the $\tcpa$ format, we have that their final marginal bids $b_i$, $b_j$ are the same.\footnote{A sufficient condition for those marginal bids to be the same is that (i) both bidders have the same target ($T_i=T_j$) and (ii) that for every query $x$ there exists a query $x'$ such that $q_{i}(x)=q_{j}(x')$.}
\end{definition}

\begin{remark}
When there is only one extra-buyer is in the auction, the game is $\tcpa$-symmetric. 
\end{remark}

We are now in position to present Theorem~\ref{th:3}.
\begin{theorem}\label{th:3}
Suppose that the auctioneer is constrained to set a uniform reserve price to the bidder and that the extra-buyers game is $\tcpa$-symmetric. Then, $u^*(\tcpa|v) >  u^*(\ua|v)$.
\end{theorem}
The proof of Theorem~\ref{th:3} is similar in spirit to that of Theorem~\ref{th:2}, but now we can not use the power of the personalized per-query reserve prices to perform the step where we``simulate world $T$ in world $M$". Instead of such a simulation, we show that in a $\tcpa$-symmetric game, there is a structural property of the bidding behavior in equilibrium, which allows us to prove the result.
We defer the proof to Appendix~\ref{app:unif-reserve}.

\section{Conclusion}
This paper attempts to explain why rational bidders (with quasi-linear utilities) choose bidding formats which at first glance are not optimal. The crux of the argument lies in the notion of lack of commitment -- the auctioneer can change (ex post) the rules of the declared auction -- once we treat the auction and bidding setting as a multi-stage game. It turns out that in a game without commitment, it is rational for the quasi-linear bidder to choose the seemingly suboptimal $\tcpa$ format over the classical $\ua$ format. We prove this in two different settings: a simpler setting with one bidder and exogenous prices, and then in the general model with endogenous prices in the auction based on bids of other buyers. In the simpler model, we also prove that the auctioneer's revenue is higher with the $\tcpa$ format in the no-commitment game compared to the $\ua$ format, under certain mild conditions. The general model requires more technically involved proofs but the insight is the same: in a world without commitment or with a lack of belief in the other player's commitment, the $\tcpa$ format aligns incentives better.  We emphasize that the core issue is in fact not the auctioneer's commitment, but the bidder's belief in the same. In the simpler model, we also provide bounds on the value of commitment, and on the loss due to the lack of belief of a bidder in a committed auctioneer. We also note that the problem of commitment may also be overcome in practice via other mechanisms such as verification in a repeated auction setting and via audits.
\bibliographystyle{ACM-Reference-Format}
\bibliography{bibliography}


\begin{thebibliography}{33}


\ifx \showCODEN    \undefined \def \showCODEN     #1{\unskip}     \fi
\ifx \showDOI      \undefined \def \showDOI       #1{#1}\fi
\ifx \showISBNx    \undefined \def \showISBNx     #1{\unskip}     \fi
\ifx \showISBNxiii \undefined \def \showISBNxiii  #1{\unskip}     \fi
\ifx \showISSN     \undefined \def \showISSN      #1{\unskip}     \fi
\ifx \showLCCN     \undefined \def \showLCCN      #1{\unskip}     \fi
\ifx \shownote     \undefined \def \shownote      #1{#1}          \fi
\ifx \showarticletitle \undefined \def \showarticletitle #1{#1}   \fi
\ifx \showURL      \undefined \def \showURL       {\relax}        \fi
\providecommand\bibfield[2]{#2}
\providecommand\bibinfo[2]{#2}
\providecommand\natexlab[1]{#1}
\providecommand\showeprint[2][]{arXiv:#2}

\bibitem[\protect\citeauthoryear{Aggarwal, Badanidiyuru, and Mehta}{Aggarwal
  et~al\mbox{.}}{2019}]%
        {AggarwalBM19}
\bibfield{author}{\bibinfo{person}{Gagan Aggarwal},
  \bibinfo{person}{Ashwinkumar Badanidiyuru}, {and} \bibinfo{person}{Aranyak
  Mehta}.} \bibinfo{year}{2019}\natexlab{}.
\newblock \showarticletitle{Autobidding with Constraints}. In
  \bibinfo{booktitle}{\emph{Web and Internet Economics - 15th International
  Conference, {WINE} 2019, New York, NY, USA, December 10-12, 2019,
  Proceedings}} \emph{(\bibinfo{series}{Lecture Notes in Computer Science})},
  Vol.~\bibinfo{volume}{11920}. \bibinfo{publisher}{Springer},
  \bibinfo{pages}{17--30}.
\newblock


\bibitem[\protect\citeauthoryear{Aggarwal, Perlroth, and Zhao}{Aggarwal
  et~al\mbox{.}}{2023}]%
        {multichannel}
\bibfield{author}{\bibinfo{person}{Gagan Aggarwal}, \bibinfo{person}{Andres
  Perlroth}, {and} \bibinfo{person}{Junyao Zhao}.}
  \bibinfo{year}{2023}\natexlab{}.
\newblock \bibinfo{title}{Multi-Channel Auction Design in the Autobidding
  World}.
\newblock
\newblock
\urldef\tempurl%
\url{https://doi.org/10.48550/ARXIV.2301.13410}
\showDOI{\tempurl}


\bibitem[\protect\citeauthoryear{Akbarpour and Li}{Akbarpour and Li}{2020}]%
        {akb}
\bibfield{author}{\bibinfo{person}{Mohammad Akbarpour} {and}
  \bibinfo{person}{Shengwu Li}.} \bibinfo{year}{2020}\natexlab{}.
\newblock \showarticletitle{Credible Auctions: A Trilemma}.
\newblock \bibinfo{journal}{\emph{Econometrica}} \bibinfo{volume}{88},
  \bibinfo{number}{2} (\bibinfo{year}{2020}), \bibinfo{pages}{425--467}.
\newblock
\urldef\tempurl%
\url{https://doi.org/10.3982/ECTA15925}
\showDOI{\tempurl}
\showeprint{https://onlinelibrary.wiley.com/doi/pdf/10.3982/ECTA15925}


\bibitem[\protect\citeauthoryear{Aumann}{Aumann}{1990}]%
        {aumann1990nash}
\bibfield{author}{\bibinfo{person}{Robert Aumann}.}
  \bibinfo{year}{1990}\natexlab{}.
\newblock \bibinfo{title}{Nash equilibria are not self-enforcing, in
  ‘‘Economic Decision-Making: Games, Econometrics and Optimization’’(JJ
  Gabszewicz, J.-F. Richard, and LA Wolsey, Eds.)}.
\newblock
\newblock


\bibitem[\protect\citeauthoryear{Balseiro, Deng, Mao, Mirrokni, and
  Zuo}{Balseiro et~al\mbox{.}}{2021a}]%
        {balseironeurips}
\bibfield{author}{\bibinfo{person}{Santiago Balseiro}, \bibinfo{person}{Yuan
  Deng}, \bibinfo{person}{Jieming Mao}, \bibinfo{person}{Vahab Mirrokni}, {and}
  \bibinfo{person}{Song Zuo}.} \bibinfo{year}{2021}\natexlab{a}.
\newblock \showarticletitle{Robust Auction Design in the Auto-bidding World}.
  In \bibinfo{booktitle}{\emph{Advances in Neural Information Processing
  Systems}}, \bibfield{editor}{\bibinfo{person}{M.~Ranzato},
  \bibinfo{person}{A.~Beygelzimer}, \bibinfo{person}{Y.~Dauphin},
  \bibinfo{person}{P.S. Liang}, {and} \bibinfo{person}{J.~Wortman Vaughan}}
  (Eds.), Vol.~\bibinfo{volume}{34}. \bibinfo{publisher}{Curran Associates,
  Inc.}, \bibinfo{pages}{17777--17788}.
\newblock
\urldef\tempurl%
\url{https://proceedings.neurips.cc/paper/2021/file/948f847055c6bf156997ce9fb59919be-Paper.pdf}
\showURL{%
\tempurl}


\bibitem[\protect\citeauthoryear{Balseiro, Deng, Mao, Mirrokni, and
  Zuo}{Balseiro et~al\mbox{.}}{2021b}]%
        {BalseiroDMMZ21}
\bibfield{author}{\bibinfo{person}{Santiago~R. Balseiro}, \bibinfo{person}{Yuan
  Deng}, \bibinfo{person}{Jieming Mao}, \bibinfo{person}{Vahab~S. Mirrokni},
  {and} \bibinfo{person}{Song Zuo}.} \bibinfo{year}{2021}\natexlab{b}.
\newblock \showarticletitle{The Landscape of Auto-bidding Auctions: Value
  versus Utility Maximization}. In \bibinfo{booktitle}{\emph{{EC} '21: The 22nd
  {ACM} Conference on Economics and Computation, Budapest, Hungary, July 18-23,
  2021}}. \bibinfo{publisher}{{ACM}}, \bibinfo{pages}{132--133}.
\newblock


\bibitem[\protect\citeauthoryear{Balseiro and Gur}{Balseiro and Gur}{2019}]%
        {BalseiroG19}
\bibfield{author}{\bibinfo{person}{Santiago~R. Balseiro} {and}
  \bibinfo{person}{Yonatan Gur}.} \bibinfo{year}{2019}\natexlab{}.
\newblock \showarticletitle{{Learning in Repeated Auctions with Budgets: Regret
  Minimization and Equilibrium}}.
\newblock \bibinfo{journal}{\emph{Management Science}} \bibinfo{volume}{65},
  \bibinfo{number}{9} (\bibinfo{date}{September} \bibinfo{year}{2019}),
  \bibinfo{pages}{3952--3968}.
\newblock


\bibitem[\protect\citeauthoryear{Bester and Strausz}{Bester and
  Strausz}{2000}]%
        {str2}
\bibfield{author}{\bibinfo{person}{Helmut Bester} {and} \bibinfo{person}{Roland
  Strausz}.} \bibinfo{year}{2000}\natexlab{}.
\newblock \showarticletitle{Imperfect commitment and the revelation principle:
  the multi-agent case}.
\newblock \bibinfo{journal}{\emph{Economics Letters}} \bibinfo{volume}{69},
  \bibinfo{number}{2} (\bibinfo{year}{2000}), \bibinfo{pages}{165--171}.
\newblock
\showISSN{0165-1765}
\urldef\tempurl%
\url{https://doi.org/10.1016/S0165-1765(00)00301-3}
\showDOI{\tempurl}


\bibitem[\protect\citeauthoryear{Bester and Strausz}{Bester and
  Strausz}{2001}]%
        {str1}
\bibfield{author}{\bibinfo{person}{Helmut Bester} {and} \bibinfo{person}{Roland
  Strausz}.} \bibinfo{year}{2001}\natexlab{}.
\newblock \showarticletitle{Contracting with Imperfect Commitment and the
  Revelation Principle: The Single Agent Case}.
\newblock \bibinfo{journal}{\emph{Econometrica}} \bibinfo{volume}{69},
  \bibinfo{number}{4} (\bibinfo{year}{2001}), \bibinfo{pages}{1077--1098}.
\newblock
\urldef\tempurl%
\url{https://doi.org/10.1111/1468-0262.00231}
\showDOI{\tempurl}
\showeprint{https://onlinelibrary.wiley.com/doi/pdf/10.1111/1468-0262.00231}


\bibitem[\protect\citeauthoryear{Bulow}{Bulow}{1982}]%
        {bulow82}
\bibfield{author}{\bibinfo{person}{Jeremy~I. Bulow}.}
  \bibinfo{year}{1982}\natexlab{}.
\newblock \showarticletitle{Durable-Goods Monopolists}.
\newblock \bibinfo{journal}{\emph{Journal of Political Economy}}
  \bibinfo{volume}{90}, \bibinfo{number}{2} (\bibinfo{year}{1982}),
  \bibinfo{pages}{314--332}.
\newblock
\showISSN{00223808, 1537534X}


\bibitem[\protect\citeauthoryear{Che and Gale}{Che and Gale}{2000}]%
        {Che2000TheOM}
\bibfield{author}{\bibinfo{person}{Yeon-Koo Che} {and} \bibinfo{person}{Ian~L.
  Gale}.} \bibinfo{year}{2000}\natexlab{}.
\newblock \showarticletitle{The Optimal Mechanism for Selling to a
  Budget-Constrained Buyer}.
\newblock \bibinfo{journal}{\emph{J. Econ. Theory}}  \bibinfo{volume}{92}
  (\bibinfo{year}{2000}), \bibinfo{pages}{198--233}.
\newblock


\bibitem[\protect\citeauthoryear{Daskalakis, Fishelson, Lucier, Syrgkanis, and
  Velusamy}{Daskalakis et~al\mbox{.}}{2020}]%
        {dask}
\bibfield{author}{\bibinfo{person}{Constantinos Daskalakis},
  \bibinfo{person}{Maxwell Fishelson}, \bibinfo{person}{Brendan Lucier},
  \bibinfo{person}{Vasilis Syrgkanis}, {and} \bibinfo{person}{Santhoshini
  Velusamy}.} \bibinfo{year}{2020}\natexlab{}.
\newblock \showarticletitle{Simple, Credible, and Approximately-Optimal
  Auctions}. In \bibinfo{booktitle}{\emph{Proceedings of the 21st ACM
  Conference on Economics and Computation}} \emph{(\bibinfo{series}{EC '20})}.
  \bibinfo{publisher}{Association for Computing Machinery},
  \bibinfo{address}{New York, NY, USA}, \bibinfo{pages}{713}.
\newblock
\showISBNx{9781450379755}
\urldef\tempurl%
\url{https://doi.org/10.1145/3391403.3399535}
\showDOI{\tempurl}


\bibitem[\protect\citeauthoryear{Deng, Mao, Mirrokni, Zhang, and Zuo}{Deng
  et~al\mbox{.}}{2022}]%
        {yuanfpa}
\bibfield{author}{\bibinfo{person}{Yuan Deng}, \bibinfo{person}{Jieming Mao},
  \bibinfo{person}{Vahab Mirrokni}, \bibinfo{person}{Hanrui Zhang}, {and}
  \bibinfo{person}{Song Zuo}.} \bibinfo{year}{2022}\natexlab{}.
\newblock \bibinfo{title}{Efficiency of the First-Price Auction in the
  Autobidding World}.
\newblock
\newblock
\urldef\tempurl%
\url{https://doi.org/10.48550/ARXIV.2208.10650}
\showDOI{\tempurl}


\bibitem[\protect\citeauthoryear{Deng, Mao, Mirrokni, and Zuo}{Deng
  et~al\mbox{.}}{2021}]%
        {DengMMZ21}
\bibfield{author}{\bibinfo{person}{Yuan Deng}, \bibinfo{person}{Jieming Mao},
  \bibinfo{person}{Vahab Mirrokni}, {and} \bibinfo{person}{Song Zuo}.}
  \bibinfo{year}{2021}\natexlab{}.
\newblock \showarticletitle{Towards Efficient Auctions in an Auto-Bidding
  World}. In \bibinfo{booktitle}{\emph{Proceedings of the Web Conference 2021}}
  \emph{(\bibinfo{series}{WWW '21})}. \bibinfo{pages}{3965–3973}.
\newblock


\bibitem[\protect\citeauthoryear{Facebook}{Facebook}{2022}]%
        {fbautobiddingsupport}
Facebook \bibinfo{year}{2022}\natexlab{}.
\newblock \bibinfo{title}{Auto-bidding~products~support~page}.
\newblock
  \bibinfo{howpublished}{\url{https://www.facebook.com/business/help/1619591734742116}}.
\newblock
\newblock
\shownote{Accessed: 2022-02-09.}


\bibitem[\protect\citeauthoryear{Feldman, Muthukrishnan, Pal, and
  Stein}{Feldman et~al\mbox{.}}{2007}]%
        {FeldmanMPS07}
\bibfield{author}{\bibinfo{person}{Jon Feldman},
  \bibinfo{person}{Shanmugavelayutham Muthukrishnan}, \bibinfo{person}{Martin
  Pal}, {and} \bibinfo{person}{Cliff Stein}.} \bibinfo{year}{2007}\natexlab{}.
\newblock \showarticletitle{Budget optimization in search-based advertising
  auctions}. In \bibinfo{booktitle}{\emph{Proceedings of the 8th ACM conference
  on Electronic commerce}}. \bibinfo{pages}{40--49}.
\newblock


\bibitem[\protect\citeauthoryear{Ferreira and Weinberg}{Ferreira and
  Weinberg}{2020}]%
        {wein}
\bibfield{author}{\bibinfo{person}{Matheus V.~X. Ferreira} {and}
  \bibinfo{person}{S.~Matthew Weinberg}.} \bibinfo{year}{2020}\natexlab{}.
\newblock \showarticletitle{Credible, Truthful, and Two-Round (Optimal)
  Auctions via Cryptographic Commitments}. In
  \bibinfo{booktitle}{\emph{Proceedings of the 21st ACM Conference on Economics
  and Computation}} \emph{(\bibinfo{series}{EC '20})}.
  \bibinfo{publisher}{Association for Computing Machinery},
  \bibinfo{address}{New York, NY, USA}, \bibinfo{pages}{683–712}.
\newblock
\showISBNx{9781450379755}
\urldef\tempurl%
\url{https://doi.org/10.1145/3391403.3399495}
\showDOI{\tempurl}


\bibitem[\protect\citeauthoryear{Fikioris and Tardos}{Fikioris and
  Tardos}{2022}]%
        {tardosarxiv}
\bibfield{author}{\bibinfo{person}{Giannis Fikioris} {and}
  \bibinfo{person}{Éva Tardos}.} \bibinfo{year}{2022}\natexlab{}.
\newblock \bibinfo{title}{Liquid Welfare guarantees for No-Regret Learning in
  Sequential Budgeted Auctions}.
\newblock
\newblock
\urldef\tempurl%
\url{https://doi.org/10.48550/ARXIV.2210.07502}
\showDOI{\tempurl}


\bibitem[\protect\citeauthoryear{Fudenberg and Tirole}{Fudenberg and
  Tirole}{1983}]%
        {fud}
\bibfield{author}{\bibinfo{person}{Drew Fudenberg} {and} \bibinfo{person}{Jean
  Tirole}.} \bibinfo{year}{1983}\natexlab{}.
\newblock \showarticletitle{{Sequential Bargaining with Incomplete
  Information}}.
\newblock \bibinfo{journal}{\emph{The Review of Economic Studies}}
  \bibinfo{volume}{50}, \bibinfo{number}{2} (\bibinfo{date}{04}
  \bibinfo{year}{1983}), \bibinfo{pages}{221--247}.
\newblock
\showISSN{0034-6527}
\urldef\tempurl%
\url{https://doi.org/10.2307/2297414}
\showDOI{\tempurl}
\showeprint{https://academic.oup.com/restud/article-pdf/50/2/221/4355731/50-2-221.pdf}


\bibitem[\protect\citeauthoryear{Gaitonde, Li, Light, Lucier, and
  Slivkins}{Gaitonde et~al\mbox{.}}{2023}]%
        {gaitondeetal23}
\bibfield{author}{\bibinfo{person}{Jason Gaitonde}, \bibinfo{person}{Yingkai
  Li}, \bibinfo{person}{Bar Light}, \bibinfo{person}{Brendan Lucier}, {and}
  \bibinfo{person}{Aleksandrs Slivkins}.} \bibinfo{year}{2023}\natexlab{}.
\newblock \showarticletitle{{Budget Pacing in Repeated Auctions: Regret and
  Efficiency Without Convergence}}. In \bibinfo{booktitle}{\emph{14th
  Innovations in Theoretical Computer Science Conference (ITCS 2023)}}
  \emph{(\bibinfo{series}{Leibniz International Proceedings in Informatics
  (LIPIcs)})}, \bibfield{editor}{\bibinfo{person}{Yael Tauman~Kalai}} (Ed.),
  Vol.~\bibinfo{volume}{251}. \bibinfo{publisher}{Schloss Dagstuhl --
  Leibniz-Zentrum f{\"u}r Informatik}, \bibinfo{address}{Dagstuhl, Germany},
  \bibinfo{pages}{52:1--52:1}.
\newblock
\showISBNx{978-3-95977-263-1}
\showISSN{1868-8969}
\urldef\tempurl%
\url{https://doi.org/10.4230/LIPIcs.ITCS.2023.52}
\showDOI{\tempurl}


\bibitem[\protect\citeauthoryear{Golrezaei, Lobel, and Paes~Leme}{Golrezaei
  et~al\mbox{.}}{2021}]%
        {GolrezaiLP21}
\bibfield{author}{\bibinfo{person}{Negin Golrezaei}, \bibinfo{person}{Ilan
  Lobel}, {and} \bibinfo{person}{Renato Paes~Leme}.}
  \bibinfo{year}{2021}\natexlab{}.
\newblock \showarticletitle{Auction Design for ROI-Constrained Buyers}. In
  \bibinfo{booktitle}{\emph{Proceedings of the Web Conference 2021}}
  \emph{(\bibinfo{series}{WWW '21})}. \bibinfo{pages}{3941–3952}.
\newblock


\bibitem[\protect\citeauthoryear{Google}{Google}{2022}]%
        {googleautobiddingsupport}
Google \bibinfo{year}{2022}\natexlab{}.
\newblock \bibinfo{title}{Auto-bidding~products~support~page}.
\newblock
  \bibinfo{howpublished}{\url{https://support.google.com/google-ads/answer/2979071}}.
\newblock
\newblock
\shownote{Accessed: 2022-02-09.}


\bibitem[\protect\citeauthoryear{Gul, Sonnenschein, and Wilson}{Gul
  et~al\mbox{.}}{1986}]%
        {gul}
\bibfield{author}{\bibinfo{person}{Faruk Gul}, \bibinfo{person}{Hugo
  Sonnenschein}, {and} \bibinfo{person}{Robert Wilson}.}
  \bibinfo{year}{1986}\natexlab{}.
\newblock \showarticletitle{Foundations of dynamic monopoly and the coase
  conjecture}.
\newblock \bibinfo{journal}{\emph{Journal of Economic Theory}}
  \bibinfo{volume}{39}, \bibinfo{number}{1} (\bibinfo{year}{1986}),
  \bibinfo{pages}{155--190}.
\newblock
\urldef\tempurl%
\url{https://EconPapers.repec.org/RePEc:eee:jetheo:v:39:y:1986:i:1:p:155-190}
\showURL{%
\tempurl}


\bibitem[\protect\citeauthoryear{Krishna}{Krishna}{2010}]%
        {krishna2010}
\bibfield{author}{\bibinfo{person}{Vijay Krishna}.}
  \bibinfo{year}{2010}\natexlab{}.
\newblock \showarticletitle{Auction Theory}.
\newblock In \bibinfo{booktitle}{\emph{Auction Theory (Second Edition)}
  (\bibinfo{edition}{second edition} ed.)}. \bibinfo{publisher}{Academic
  Press}, \bibinfo{address}{San Diego}, \bibinfo{pages}{iii}.
\newblock


\bibitem[\protect\citeauthoryear{Laffont and Robert}{Laffont and
  Robert}{1996}]%
        {Laffont96optimal}
\bibfield{author}{\bibinfo{person}{Jean-Jacques Laffont} {and}
  \bibinfo{person}{Jacques Robert}.} \bibinfo{year}{1996}\natexlab{}.
\newblock \showarticletitle{Optimal auction with financially constrained
  buyers}.
\newblock \bibinfo{journal}{\emph{Economics Letters}} \bibinfo{volume}{52},
  \bibinfo{number}{2} (\bibinfo{year}{1996}), \bibinfo{pages}{181--186}.
\newblock


\bibitem[\protect\citeauthoryear{Liaw, Mehta, and Perlroth}{Liaw
  et~al\mbox{.}}{2022}]%
        {lmp22}
\bibfield{author}{\bibinfo{person}{Christopher Liaw}, \bibinfo{person}{Aranyak
  Mehta}, {and} \bibinfo{person}{Andres Perlroth}.}
  \bibinfo{year}{2022}\natexlab{}.
\newblock \bibinfo{title}{Efficiency of non-truthful auctions under
  auto-bidding}.
\newblock
\newblock
\urldef\tempurl%
\url{https://doi.org/10.48550/ARXIV.2207.03630}
\showDOI{\tempurl}


\bibitem[\protect\citeauthoryear{Liu, Mierendorff, Shi, and Zhong}{Liu
  et~al\mbox{.}}{2019}]%
        {liu}
\bibfield{author}{\bibinfo{person}{Qingmin Liu}, \bibinfo{person}{Konrad
  Mierendorff}, \bibinfo{person}{Xianwen Shi}, {and} \bibinfo{person}{Weijie
  Zhong}.} \bibinfo{year}{2019}\natexlab{}.
\newblock \showarticletitle{Auctions with Limited Commitment}.
\newblock \bibinfo{journal}{\emph{American Economic Review}}
  \bibinfo{volume}{109}, \bibinfo{number}{3} (\bibinfo{date}{March}
  \bibinfo{year}{2019}), \bibinfo{pages}{876--910}.
\newblock
\urldef\tempurl%
\url{https://doi.org/10.1257/aer.20170882}
\showDOI{\tempurl}


\bibitem[\protect\citeauthoryear{Mehta}{Mehta}{2022}]%
        {mehta22}
\bibfield{author}{\bibinfo{person}{Aranyak Mehta}.}
  \bibinfo{year}{2022}\natexlab{}.
\newblock \showarticletitle{Auction Design in an Auto-Bidding Setting:
  Randomization Improves Efficiency Beyond VCG}. In
  \bibinfo{booktitle}{\emph{Proceedings of the ACM Web Conference 2022}}
  \emph{(\bibinfo{series}{WWW '22})}. \bibinfo{publisher}{Association for
  Computing Machinery}, \bibinfo{address}{New York, NY, USA},
  \bibinfo{pages}{173–181}.
\newblock
\showISBNx{9781450390965}
\urldef\tempurl%
\url{https://doi.org/10.1145/3485447.3512062}
\showDOI{\tempurl}


\bibitem[\protect\citeauthoryear{Milgrom and Segal}{Milgrom and Segal}{2002}]%
        {milg02}
\bibfield{author}{\bibinfo{person}{Paul Milgrom} {and} \bibinfo{person}{Ilya
  Segal}.} \bibinfo{year}{2002}\natexlab{}.
\newblock \showarticletitle{Envelope Theorems for Arbitrary Choice Sets}.
\newblock \bibinfo{journal}{\emph{Econometrica}} \bibinfo{volume}{70},
  \bibinfo{number}{2} (\bibinfo{year}{2002}), \bibinfo{pages}{583--601}.
\newblock


\bibitem[\protect\citeauthoryear{Myerson}{Myerson}{1981}]%
        {myerson1981optimal}
\bibfield{author}{\bibinfo{person}{Roger~B Myerson}.}
  \bibinfo{year}{1981}\natexlab{}.
\newblock \showarticletitle{Optimal auction design}.
\newblock \bibinfo{journal}{\emph{Mathematics of operations research}}
  \bibinfo{volume}{6}, \bibinfo{number}{1} (\bibinfo{year}{1981}),
  \bibinfo{pages}{58--73}.
\newblock


\bibitem[\protect\citeauthoryear{Rudin}{Rudin}{1987}]%
        {rudin}
\bibfield{author}{\bibinfo{person}{Walter Rudin}.}
  \bibinfo{year}{1987}\natexlab{}.
\newblock \bibinfo{booktitle}{\emph{Real and Complex Analysis, 3rd Ed.}}
\newblock \bibinfo{publisher}{McGraw-Hill, Inc.}, \bibinfo{address}{USA}.
\newblock
\showISBNx{0070542341}


\bibitem[\protect\citeauthoryear{Skreta}{Skreta}{2006}]%
        {skr2}
\bibfield{author}{\bibinfo{person}{Vasiliki Skreta}.}
  \bibinfo{year}{2006}\natexlab{}.
\newblock \showarticletitle{Sequentially Optimal Mechanisms}.
\newblock \bibinfo{journal}{\emph{The Review of Economic Studies}}
  \bibinfo{volume}{73}, \bibinfo{number}{4} (\bibinfo{year}{2006}),
  \bibinfo{pages}{1085--1111}.
\newblock
\showISSN{00346527, 1467937X}
\urldef\tempurl%
\url{http://www.jstor.org/stable/4123260}
\showURL{%
\tempurl}


\bibitem[\protect\citeauthoryear{Skreta}{Skreta}{2015}]%
        {skr1}
\bibfield{author}{\bibinfo{person}{Vasiliki Skreta}.}
  \bibinfo{year}{2015}\natexlab{}.
\newblock \showarticletitle{Optimal auction design under non-commitment}.
\newblock \bibinfo{journal}{\emph{Journal of Economic Theory}}
  \bibinfo{volume}{159} (\bibinfo{year}{2015}), \bibinfo{pages}{854--890}.
\newblock
\showISSN{0022-0531}
\urldef\tempurl%
\url{https://doi.org/10.1016/j.jet.2015.04.007}
\showDOI{\tempurl}
\newblock
\shownote{Symposium Issue on Dynamic Contracts and Mechanism Design.}


\end{thebibliography}

\clearpage
\appendix
\section{Missing Proofs from Section~\ref{sec:model}}

\begin{proof}[Proof of Lemma~\ref{lem:1}]
First, notice that Problem~\eqref{eq:0} has a unique solution since $H$ is an increasing function (Assumption~\ref{ass:0}). 

Second, for existence of solution, observe that $b=T$ is feasible solution. This implies that the optimal solution has $b(T;r)\geq T$. Also, by writing the constraint as function of $b$, $g(b) = r H(r) + \int_{r}^b zh(z)dz - TH(b)$ we have that $g$ is increasing on $b>T$ since $g'(b) = (b-T)h(b)$ which is strictly by assumption on $H$. Next, using integration by parts on the integral term in $g$, we rewrite $g(b) = (b-T)H(b) - \int_{r}^b H(z)dz$. Because $\int_{r}^{\infty} H(z)dz<\infty$ (see Assumption~\ref{ass:0}), we have that $\lim_{b\to \infty} g(b) =\infty$. We conclude that Problem~\eqref{eq:0} has a unique solution satisfying $g(b(T;r))=0$.

The previous paragraph shows that Equation~\eqref{eq:1} has a unique solution for $b\in [T,\infty)$ which is $b(T;r)$. For $b<T$, observe that the left-hand-side of Equation~\eqref{eq:1} is negative while the right-hand-side is positive. Therefore, Equation~\eqref{eq:1} has a unique solution for $b\geq r$.

To conclude the proof of the lemma observe that $\partial_T g(b|T) = - H(b) <0$. This implies that for $T<T'$ $g(b(T;r)|T')<g(b(T;r)|T)= 0$. Because $g(\cdot|T')$ is increasing, we conclude that $b(T';r)>b(T;r)$. The same logic holds to show that $b(T;r)$ is decreasing in $r$ since $\partial_r g(b|r)= H(r)>0$.  
\end{proof}

\section{Missing Proofs from Section~\ref{sec:3}}\label{app:simple}

\begin{proof}[Proof of Proposition~\ref{prop:1}]
For the proof we denote by $\underline v, \overline v$ the lowest and upper elements in the support of $F$.

From the revelation principle, without loss of generality we restrict to incentive compatible mechanisms $(\chi(v),\tau(v))$.

The incentive compatible condition on the mechanism implies that the bidder's utility satisfies that
$$ vH(\chi(v)) - \tau(v) = \max_{v'} vH(\chi(v')) - \tau(v').$$

Using the envelope approach \citep{milg02}, we have that the incentive compatibility restriction is equivalent to (i) $\chi(v)$ being non-decreasing (since $H$ is increasing) and (ii) that for almost every type $v$ we have that 
$$\tau(v) = vH(\chi(v)) - \int_{\underline v}^v H(\chi(z))dz \qquad \mbox{ (IC)}.$$

The previous characterization of $\tau$ allow us to reformulate the auctioneer's revenue as
\begin{align*}\int_{\underline v}^{\overline v} \left(vH(\chi(v)) - \int_{\underline v}^v H(\chi(z))dz\right) & f(v)dv =\\ 
&\int_{\underline v}^{\overline v} \phi_F(v) H(\chi(v)) f(v)dv,\end{align*}
where the equality comes from using integration by parts.\footnote{The integrability condition on $H$ (Assumption~\ref{ass:0}) ensures that the integration by parts can be taken for the case where the support of the distribution $F$ is unbounded.}

The intrinsic price condition imposes that for every query, the price which is sold has to be at least the respective intrinsic price. This, in terms of $H$, implies that for every feasible transfer rule $\tau$ and every valuation-types $v>v'$ the following inequality holds:
$$ \tau(v)- \tau(v') \geq \int_{\chi(v')}^{\chi(v)}zh(z)dz \qquad \mbox{ (FC)}.$$

We claim that if $\chi$ satisfies the incentive-compatibility constraint (IC) and the feasibility constraint (FC) then, up to a zero-measure set of types, $\chi(v)\leq v$. To see this, since $\chi$ is non-decreasing then is differentiable almost everywhere (see Theorem 7.20 \cite{rudin}). Consequently, $\tau$ is also differentiable almost everywhere. The (IC) condition implies that $\tau'(v) = vh(\chi(v))\chi'(v)$, while the (FC) condition implies that $\tau'(v) \geq \chi(v) h(\chi(v))\chi'(v)$. We conclude that $\chi(v)\leq v$ for almost every valuation type $v$.    

To conclude the proof, consider the following relaxation of the auctioneer's problem 
\begin{align*}
    \max_{\chi(\cdot)}\quad &\E_v \left[H(\chi(v)) \;\phi_F(v)\right] \\
    \mbox{s.t. }\quad & \chi(v)\leq v \quad (a.e.).
\end{align*}

Because $H$ is increasing and $\phi_F$ is increasing, the above problem admits a pointwise solution which is 
$$ \chi^*(v)= \begin{cases}
0 &\mbox{ if } v < r_{\mye} \\
v &\mbox{ if } v\geq r_{\mye}
\end{cases}.
$$

We conclude the proof by showing that $\chi^*$ is solution for the auctioneer's problem. Indeed, first observe that $\chi^*$ is monotone, and hence, $\chi^*$ satisfies condition (i) in the incentive-compatibility restriction. Moreover, the transfer rule generated by $\chi^*$ is such that for $v\geq r_{\mye}$,
$$\tau^*(v) = vH(v) - \int_{r_{\mye}}^v H(z)dz = r_{\mye} H(r_{\mye}) +\int_{r_{\mye}}^v z h(z)dz, $$
where we use integration by parts for the second equality. Clearly, $\tau^*$ satisfies the feasibility constraint (FC). We conclude that $\chi^*$ is a feasible solution for the auctioneer's problem, and therefore, $(\chi^*,\tau^*)$ is the revenue-maximizing mechanism. 
\end{proof}

\begin{proof}[Proof of Proposition~\ref{prop:ce}]
From Equation~\eqref{eq:tcpa}, we have that $T^*(v) = v- \frac{\int_0^v \hat H(z)dz}{\hat H(v)} = \frac{\log(v)}{\hat H(v)}$. Thus, $\pi^*(\tcpa|v) = T^*(v) \hat H(v) = \log(v)$. Similarly, the first order conditions in Problem~\eqref{eq:fpa} implies that $b^*(v)$ solves
$$ (v-b^*(v))\hat h(b^*(v))= \hat H(b^*(v)).$$
Replacing $\hat H$, we obtain that the unique solution is $b^*(v)=\sqrt{v}$. We conclude since $\pi^*(\ua|v)= b^*(v)H(b^*(v)) = \sqrt{v}-1$. 

Observe that for every $v\geq 1$, we have that $\pi^*(\tcpa|v)>\pi^*(\ua|v)$ and $\lim_{v\to \infty} \frac{\pi^*(\tcpa|v)}{\pi^*(\ua|v)}=\infty$. Therefore for every $\gamma>0$, we can find a $\underline v>0$ large enough so that if all types in the support of the distribution $F$ are greater than $\underline v$, we have that $\pi^*({\ua})>\gamma \cdot\pi^*({\tcpa})$.
\end{proof}

\begin{proof}[Proof of Proposition~\ref{prop:vc}]
Let $\underline v, \overline v$ the smallest and largest elements in the support of $F$. Using that $\pi^*_{NC}= \pi^*({\tcpa})$, we get that
\begin{align*}\pi^*_{NC}&= \E_{v}[T^*(v) H(v)]\\
&= \int_{\underline v}^{\overline v} \left(vH(v) - \int_0^v H(z) dz\right) f(v)dv\\
&= \int_0^{\underline v} H(z)dz + \E_{v}[\phi(v)H(v)]\\
&\geq \underbrace{\E_{v}[\phi(v)H(v)]}.
\end{align*}
The second equality comes from the characterization of $T^*$ in Equation~\eqref{eq:tcpa}. The thrid equality is a simple integration by parts argument. The inequality last holds because $H$ is nonnegative.

On the other hand, using integration by parts in the characterization of $\pi^*_C$ on Proposition~\ref{prop:2}, we get that
$$\pi^*_C =  \E_{v}[\mathbf{1}_{\{v\geq r_{\mye}\} } \phi(v)H(v)].$$

Therefore, we conclude that 
\begin{align*}
     \frac{\pi^*_{NC}}{\pi^*_{C}} &= 1 + \frac{\E_{v}[\mathbf{1}_{\{v< r_{\mye}\} }\phi_F(v)H(v)]}{\E_{v}[\mathbf{1}_{\{v\geq r_{\mye}\} }\phi_F(v)H(v)]}\\
     &\geq  1 + \frac{ H(r_{\mye}) \E[\mathbf{1}_{\{v< r_{\mye}\} }\phi_F(v)]}{H(r_{\mye})  \E_{v}[\mathbf{1}_{\{v\geq r_{\mye}\} }\phi_F(v)]}\\
     &\geq 1 + \frac{ \E[\mathbf{1}_{\{v< r_{\mye}\} }\phi_F(v)]}{  \E_{v}[\mathbf{1}_{\{v\geq r_{\mye}\} }\phi_F(v)]}\\
     & = 1 + \frac{\pi_{\underline v\mbox{-}\spa} - \pi_{\mye}}{\pi_{\mye}} \\
     &= \psi
     \end{align*}.
The first two inequalities hold because $\E_{v}[\mathbf{1}_{\{v\geq r_{\mye}\} }\phi_F(v)]<0$ and $H$ is positive and increasing. The second to last equality comes from the fact that $\pi_{\underline v\mbox{-}\spa} = \E[\phi_F(v)]$.     
     
To show the tightness of our result, consider $H_\epsilon(z) = \mathbf{1}_{\{z\geq \underline v\}}+ \epsilon z$. As $\epsilon\to 0$, we have that $\pi^*_{NC}\to \E[\phi(v)]=\pi_{\underline v\mbox{-}\spa}$ while $\pi^*_C\to\pi_{\mye}$.
\end{proof}

\begin{proof}[Proof of Proposition~\ref{prop:pr}]
For $\epsilon>0$, consider the following instance. $F(v)= v$ (uniform distribution) and
\begin{equation}\label{eq:h_eps}
H_\epsilon(p) = \begin{cases}
\frac{1-\epsilon}\epsilon p &\mbox { if } p\leq \epsilon\\
(1-\epsilon) +\epsilon p &\mbox{ if } p\in[\epsilon, 1]
\end{cases}.
\end{equation}
From Proposition~\ref{prop:1}, we have that $r_{\mye} = 1/2$. Also, from Equation~\eqref{eq:tcpa} we have that the optimal target $T^*_\epsilon(1)\leq \epsilon$. Hence for $\epsilon$ small enough, we have that $T^*_\epsilon(1)<r_{\mye}$, and thus, for every $v\in [0,1]$ we have that $T^*_\epsilon(v)< T^*_\epsilon(1)<r_{\mye}$ since $T^*(v)$ is increasing in $v$ (see Lemma~\ref{lem:1}). Therefore, $\pi^*_{WB}=0$. On the other hand, $\pi^*_{NC}\geq  r_{\mye} (1-F(r_{\mye})) (1-\epsilon)$. We conclude that for $\epsilon$ small enough $\pi^*_{C} > \gamma \cdot \pi^*_{WB}$, independently of the factor $\gamma$.
\end{proof}

\label{app:symmetric}

\section{Proof of Theorem~\ref{th:3}}\label{app:unif-reserve}

Before proving the theorem, we require the following technical lemma.

\setcounter{lemma}{3}
\begin{lemma}[No Swapping Lemma]\label{lem:swap}
Consider two final marginal bids profile $(b_0,\b)$ and $(b'_0,\b')$ and a particular realization $\omega_v$. Consider a directed graph $G=(V,E)$, where each vertex $v$ is one of the participants on the auction and the edge $e=(i,j)$ exists if and only if for profile $(b_0,\b)$ agent $i$ wins query $x$ and for the profile $(b'_0,\b')$ agent $j$ wins query $x$. Then the graph $G$ is acyclical. 
\end{lemma}

\begin{proof}
Suppose that for some $\omega_v$, $G$ contains a cycle.

Then consider the sequence of bidders in the auction $i_1,\ldots, i_k, i_1$ creating the cycle. Let $x_{i_j}$ the query that $i_j$ wins with the bid profile $(b_0,\b)$ and that $i_{j+1}$ wins when the bid profile is $(b'_0,\b')$. Observe that $b_{i_j}q_{i_j}(x_{i_j})>b_{i_{j+1}}q_{i_{j+1}}(x_{i_j})$. Therefore, we derive that
\begin{equation}\label{eq:aux} b_{i_1} > b_{i_{2}} \frac{q_{i_2}(x_{i_1})}{q_{i_1}(x_{i_1})} > b_{i_1} \prod_{j=1}^k \frac{q_{i_{j+1}}(x_{i_j})}{q_{i_{j}}(x_{i_j})}.
\end{equation}
where $k+1 = i$.

Using the same logic for the bidding profile $(b'_0,\b')$ we obtain that
\begin{equation}\label{eq:aux2}
b'_{i_1} < b'_{i_{2}} \frac{q_{i_2}(x_{i_1})}{q_{i_1}(x_{i_1})} < b'_{i_1} \prod_{j=1}^k \frac{q_{i_{j+1}}(x_{i_j})}{q_{i_{j}}(x_{i_j})}.
\end{equation}

Equations \eqref{eq:aux} and \eqref{eq:aux2} generate a contradiction. We conclude that $G$ does not have any cycle.
\end{proof}

We are now in a position to prove Theorem~\ref{th:3}. 

\begin{proof}[Proof of Theorem~\ref{th:3}]
The proof strategy is similar to Theorem~3. Given an arbitrary marginal bid $b_{\tiny \ua}$, we want to show that if the bidder submits $b_{\tiny \ua}\in [0,v]$ using the $\ua$ format, the bidder can weakly improves her payoff by bidding a target $T=b_{\tiny \ua}$ with the $\tcpa$ format for every realization $\omega_v$. Furthermore, the inequality is strict for a positive measure of $\omega_v$. We split our proof into the following steps.  \medskip

Let $X^0_{\tiny \ua}(
\omega_v)$ and $X^0_{\tiny \ua}(
\omega_v)$ the subset of queries that the bidder obtains in the $\ua$ and $\tcpa$ cases, respectively. We assert that  $X^0_{\tiny \ua}(\omega_v) \subseteq X^0_{\tiny \tcpa}(\omega_v)$.  Suppose for the sake of a contradiction that is not true. Then, consider $x \in X^0_{\tiny \ua}(\omega_v) \setminus X^0_{\tiny \ua}(\omega_v)$. Let $b_{\tiny \ua}$ and $b'_{\tiny \ua}$ the bidder's final marginal bid for the $\ua$ and $\tcpa$ cases, respectively. From Lemma~1 we have that $b'_{\tiny \ua}\geq b_{\tiny \ua}$. Therefore, since the bidder is losing query $x$ in the $\tcpa$ case an extra-buyer $i$ has increased her final marginal bid from $b_i$ to $b'_i$ and wins query $x$. Thus, for the graph $G$ described in Lemma~\ref{lem:swap} we have an edge $e=(0,i)$.

Because the marginal bid of Extra-Buyer $i$ is not constant, it implies that she is using the $\tcpa$ format. We assert that there is a query $x'\in X^0_{\tiny \ua}(\omega_v) \setminus X^0_{\tiny \tcpa}(\omega_v)$. Because $\ua$ extra-buyers do not change their marginal bids and $\tcpa$-extra buyers have a higher marginal bid in the $\tcpa$ case (since $b'_i > b_i$ and the game of $\tcpa$-buyers is symmetric), the cost for every query is (weakly) higher in the $\tcpa$ case than in the $\ua$ case. In particular, the average cost-per-acquisition (CPA) of wining queries $X^0_{\tiny \ua}(\omega_v)$ is at least $T_i$. Moreover, because the queries in $X\setminus X^0_{\tiny \ua}(\omega_v)$ have a CPA higher than the $T_i$ (since $b_i\geq T_i$ but the Extra-Buyer $i$ did not win those queries in the $\ua$ case), we have that in order to win query $x$ and, at the same time, keep an average CPA no more than $T_i$, Extra-Buyer $i$ is loosing a query $x'\in X^0_{\tiny \ua}(\omega_v) \setminus X^0_{\tiny \tcpa}(\omega_v)$. If the winner of that query is the bidder, then the edge $e= (i,0)$ belongs to $G$. From Lemma~\ref{lem:swap} this is not possible. Suppose the winner of the query is a different Extra-Buyer $k$. In that case, we can reiterate the logic of this paragraph and either obtain a cycle on $G$ or, again, find another extra buyer winning a new query. Because there is a finite set of extra-buyers, at some point in the iteration $G$ will have a cycle. This is a contradiction. Therefore, $X^0_{\tiny \ua}(\omega_v) \subseteq X^0_{\tiny \tcpa}(\omega_v)$.

To finish the proof, we follow the same reasoning as the one used for the proof of Theorem~3. We use Step 3. to show the weak inequality for every $\omega_v$ and Step 4 to show that the inequality is strict for a positive measure of $\omega_v$. 
\end{proof}

\end{document}